\documentclass[12pt]{article}

\usepackage[margin=1in]{geometry}  
\usepackage{epsfig}
\usepackage{graphics,graphicx,color}
\usepackage{amssymb,amsfonts,amsthm,amsmath}
\usepackage{multirow, booktabs}
\usepackage{mathrsfs}
\usepackage{natbib, hhline}
\usepackage{bm}

\newtheorem{lemma}{Lemma}
\newtheorem{definition}{Definition}
\newtheorem{proposition}{Proposition}
\newtheorem{theorem}{Theorem}
\newtheorem{algorithm}{Algorithm}

\newcommand{\tr}{\mbox{tr}}
\DeclareMathOperator{\E}{\mathbb{E}}

\begin{document}
\begin{center}
{\Large \bf A Generative Approach to Joint Modeling of Quantitative and Qualitative Responses}\\
Xiaoning Kang$^a$, Lulu Kang$^b$, Wei Chen$^c$ and Xinwei Deng$^d$\footnote{Address for correspondence: Xinwei Deng, Associate Professor, Department of Statistics, Virginia Tech, Blacksburg, VA, 24061 (E-mail: xdeng@vt.edu).} \\
\vskip 5pt
{$^a$International Business College and Institute of Supply Chain Analytics, \\
Dongbei University of Finance and Economics, Dalian, China} \\
{$^b$Department of Applied Mathematics, Illinois Institute of Technology, Chicago, USA} \\
{$^c$Department of Mechanical, Materials \& Aerospace Engineering, \\
Illinois Institute of Technology, Chicago, USA}\\
$^d$Department of Statistics, Virginia Tech, Blacksburg, USA
\end{center}

\begin{abstract}
In many scientific areas, data with quantitative and qualitative (QQ) responses are commonly encountered with a large number of predictors.
By exploring the association between QQ responses, existing approaches often consider a joint model of QQ responses given the predictor variables.
However, the dependency among predictive variables also provides useful information for modeling QQ responses.
In this work, we propose a generative approach to model the joint distribution of the QQ responses and predictors.
The proposed generative model provides efficient parameter estimation under a penalized likelihood framework.
It achieves accurate classification for qualitative response and accurate prediction for quantitative response with efficient computation.
Because of the generative approach framework, the asymptotic optimality of classification and prediction of the proposed method can be established under some regularity conditions.
The performance of the proposed method is examined through simulations and real case studies in material science and genetics.
\end{abstract}


\section{Introduction}
Analyzing data with heterogeneous types of responses has been an important topic with broad applications.
Such heterogeneous data often involve both quantitative and qualitative (QQ) responses.
For example,
\cite{klein2019mixed} described a human health study examining the risk factors of adverse birth outcomes, which contains a qualitative response ``presence/absence of low birth weight" and a quantitative response ``gestational age".
In material science, the properties of a material are often characterized by QQ measures.
As shown in the case study of Section \ref{sec:realdata} on Heusler compounds, two metrics, the \emph{mixing enthalpy} (quantitative) and the \emph{global stability based on hull energy} (qualitative) are used to determine the thermodynamic stability of a full Heusler compound.
The QQ responses is a special case of ``mixed outcomes'' in the literature.
In this paper, we focus on two types of mixed responses: the quantitative continuous response and the multi-class qualitative response.

In the literature of mixed outcomes, particularly on QQ responses, it has been known that overlooking the relationship between the QQ responses is inappropriate.
Researches on the joint model of the QQ responses include early ones such as \cite{fitzmaurice1995regression, moustaki2000generalized, dunson2000bayesian, gueorguieva2001correlated, dunson2003dynamic} and recent ones such as \cite{deng2015qq, kurum2016time, kang2018bayesian, amini2018longitudinal, klein2019mixed}.
These works mostly consider a joint regression model conditioned on the predictor variables.
Based on our best knowledge, we can group them into two categories.

The first group of methods considers a conditional regression model, in which one type of the QQ responses is treated as the response of the model, and the other type is treated as the regressor.
For instance, \cite{fitzmaurice1997regression} introduced a marginal regression model of quantitative response conditioned on the qualitative response.
\cite{song2009joint} used Gaussian copulas to integrate separate one-dimensional generalized linear models into a joint regression model for mixed outcomes.
\cite{lin2010association} developed a conditional mixed-effects model to analyze clustered data containing QQ responses.
These methods are suitable for data with a small number of predictor variables.
To handle the high-dimensional input variables, \cite{deng2015qq} proposed a conditional model that encourages model sparsity through a constrained likelihood estimation.
However, inferences and asymptotic properties of their method have not been explored due to the complicated constrained likelihood estimation.
\cite{kang2018bayesian} considered a Bayesian estimation for the conditional model of \cite{deng2015qq} to obtain proper inferences of model parameters.
Nevertheless, their work is not designated for studying the asymptotic properties of the proposed estimator.
More related works can be found in \cite{chen2014selection, yang2014mixed, guglielmi2018a}, among others.

The second group of methods considers a continuous latent variable for the qualitative response, and then jointly models the latent variable and the quantitative response \citep{sammel1997latent, dunson2005bayesian, bello2012hierarchical}.
For example, \cite{gueorguieva2001correlated} studied a probit model with a latent variable and developed a Monte Carlo expectation-conditional maximization algorithm for parameter estimation.
\cite{klein2019mixed} introduced the idea of latent variable into the framework of copula regressions, constructing a latent continuous representation of binary regression models.
However, the use of latent variables often involves considerable computation in the parameter estimation.
It also makes the investigation of theoretical properties difficult.
Moreover, most of these works focus on the binary qualitative response and their model assumptions may not be easily extended to the multi-class qualitative response cases.

In this work, we propose a novel approach to jointly model the QQ responses based on the \emph{generative approach}. The proposed generative model considers the joint distribution of the high-dimensional input variables, the quantitative responses, and the multi-class qualitative response.
It is a very unique and different perspective from the existing literature and also brings advantages in both theoretical and computational aspects.
The proposed method can accommodate multi-class qualitative response and multivariate quantitative responses with attractive theoretical properties.
We call the proposed method \emph{GAQQ}, a Generative Approach for QQ responses.

The key contributions of this work are summarized as follows.
First, based on the generative model framework, we are able to establish the asymptotic properties of the proposed estimators with respect to both the classification accuracy of the qualitative response and the prediction accuracy of the quantitative response under some regularity conditions.
Such conditions are commonly used in the regularized estimation framework \citep{Shao2011Sparse, zhao2006on}.
The classification of the qualitative response enjoys the asymptotic optimality of the resulting linear discriminate classification rule.
The mean squared error (MSE) of prediction for the quantitative response is as good as the optimal prediction under the Bayes risk.
Second, an efficient procedure for parameter estimation is developed via the regularized log-likelihood function of the joint distribution of input variables and QQ responses.
Specifically, we impose regularization on both the mean differences and the covariance matrix from the joint distribution to achieve sparsity for high-dimensional predictor variables.
Third, the use of the generative approach leads to an effective prediction procedure by inferring the conditional distribution of QQ responses conditioned on the predictor variables.
That is, the quantitative response is predicted through the property of conditional multivariate normal distribution,
and the linear discriminant analysis (LDA) is employed for classification of the qualitative response.
Fourth, the proposed generative model allows the parameters related to QQ responses to be mutually learned from each other,
which is different from existing methods in which only modeling one type of QQ responses attempts to benefit from the information of the other type of QQ responses.

The remainder of this paper is organized as follows.
Section \ref{sec:methodology} details the proposed method.
The main theoretical results are presented in Section \ref{sec:theory}.
Simulation and real data analysis are conducted in Sections \ref{sec:simulation} and \ref{sec:realdata}, respectively.
Section \ref{sec:discussion} concludes this work with some discussion. Technical proofs are in the Appendix.

\section{The Proposed GAQQ Method}\label{sec:methodology}
\subsection{The Proposed Model}\label{subsec:derivation}
Suppose that the variables of interest are denoted by $(\bm X, y, Z)$ where $\bm X = (X_1, \ldots, X_{p-1})'$ is a $(p-1)$ dimensional vector of predictor variables, $y$ is a quantitative response variable and $Z \in \{1, 2\}$ is a qualitative response variable.
From a generative modeling perspective, we consider the data generation mechanism as $p(\bm X, y, Z) = p(\bm X, y|Z)p(Z)$,
indicating that data are from two classes $G_1$ and $G_2$ under $(\bm X, y)|Z$.
Assume that $\bm W = (\bm X', y)'$ follows multivariate norm distributions with different means for two classes, but sharing the same covariance matrix as follows
\begin{equation}\label{assumption}
G_1: \bm W|Z = 1 \sim N(\bm \mu_1, \bm \Sigma), ~~~ G_2: \bm W|Z = 2 \sim N(\bm \mu_2, \bm \Sigma).
\end{equation}
Denote the observed data $\bm w_1, \cdots, \bm w_{n_1}, \bm w_{n_1+1}$, $\cdots,\bm w_{n_1+n_2}$ with the first $n_1$ observations from $G_1$ and the rest $n_2$ observations from $G_2$, where $\bm w_i = (\bm x_i', y_i)', i = 1, 2, \ldots, n_1 + n_2$. Let $n=n_1+n_2$. The log-likelihood function can be written as
\begin{align}\label{eq: likelihood}
L(\bm \mu_1, \bm \mu_2, \bm \Sigma) =  n\ln|\bm C| - \sum_{k=1}^2 \sum_{i\in G_k} (\bm w_i - \bm \mu_k)' \bm C (\bm w_i - \bm \mu_k),
\end{align}
up to some constant, where
$\bm C=\bm \Sigma^{-1}$ is the inverse covariance matrix.
Let $\pi_1$ and $\pi_2$ be the prior probability of $\bm w$ belonging to classes $G_1$ and $G_2$, respectively.
Hence, the LDA assigns a new observation $\bm w$ to $G_1$ if
\begin{equation}\label{eq: rule}
\ln\frac{\hbox{Pr}(G_1|\bm W=\bm w)}{\hbox{Pr}(G_2|\bm W=\bm w)}=\ln\frac{\pi_1}{\pi_2}-\frac{1}{2}(\bm \mu_1+\bm \mu_2)' \bm C \bm \delta + \bm w' \bm C \bm \delta \geq 0,
\end{equation}
where $\bm \delta = \bm \mu_1-\bm \mu_2$.
Otherwise, $\bm w$ is classified to $G_2$.
The estimates of $\pi_1$ and $\pi_2$ are the empirical proportions of data from each class.
The parameters $\bm \mu_1,\bm \mu_2$ and $\bm C$ can be estimated by maximizing the log-likelihood function of \eqref{eq: likelihood}.

For high-dimensional data when $p \geq n$, the regularization is often needed to ensure the proper estimation of inverse covariance matrix $\bm C$ and mean difference $\bm \delta$.
We thus propose to penalize $\bm C = (c_{ij})_{1 \leq i, j \leq p}$ and $\bm \delta$ simultaneously, resulting in the following optimization problem
\begin{align}\label{eq: 1}
\min_{(\bm \mu_1,\bm \mu_2,\bm C)} -n\ln|\bm C| + \sum_{k=1}^2\sum_{i\in G_k}(\bm w_i-\bm \mu_k)' \bm C (\bm w_i-\bm \mu_k)+\lambda_1||\bm C||_1 + \frac{1}{2} \lambda_2|\bm \mu_1-\bm \mu_2|_1,
\end{align}
where $||\bm C ||_1 = \sum_{i \neq j}  |c_{ij}|$, and
$|\bm \alpha|_1 = \mathop{\sum}\limits_i |\alpha_i|$ with $\alpha_i$ being the $i$th entry of vector $\bm \alpha$.
Here $\lambda_1\ge 0$ and $\lambda_2\ge 0$ are two tuning parameters.
By applying such regularization, the proposed model can encourage the sparse structures in $\bm C$ and $\bm \delta$ at the same time.
Note that similar spirits of regularizing both $\bm C$ and $\bm \delta$ are used in several works on the LDA \citep{Shao2011Sparse, cai2012a}.

To estimate the parameters, we develop an iterative procedure to solve the sub-optimization problem with respect to $\bm C$ and $\bm \delta$ respectively.
Define $\bm \delta_2=(\bm \mu_1-\bm \mu_2)/2$ as well as $\bm \gamma=(\bm \mu_1+\bm \mu_2)/2$,
then accordingly we have $\bm \mu_1=\bm \delta_2+\bm \gamma$ and $\bm \mu_2=\bm \gamma-\bm \delta_2$.
As a result, the optimization problem \eqref{eq: 1} is re-written as
\begin{align}\label{eq: like}
\min_{(\bm \delta_2,\bm \gamma,\bm C)} &-n\ln|\bm C|+\sum_{i\in G_1}(\bm w_i-\bm \delta_2-\bm \gamma)'\bm C(\bm w_i-\bm \delta_2-\bm \gamma)\nonumber\\
&+\sum_{i\in G_2}(\bm w_i+\bm \delta_2-\bm \gamma)'\bm C(\bm w_i+\bm \delta_2-\bm \gamma)+\lambda_1||\bm C||_1+\lambda_2|\bm \delta_2|_1.
\end{align}
It is thus easy to obtain the maximum likelihood estimate of $\bm \gamma$ from (\ref{eq: like}) as
\begin{equation}\label{eq: gamma}
\hat{\bm \gamma}=\bar{\bm w}+\frac{n_2-n_1}{n}\bm \delta_2,
\end{equation}
where $\bar{\bm w}=\frac{1}{n}\sum_{i=1}^{n} \bm w_i$ is the overall mean. Then plugging $\hat{\bm \gamma}$ back into (\ref{eq: like}) yields
\begin{align}\label{eq: obj}
&(\hat{\bm \delta}_2, \hat{\bm C}) = \arg\min_{\bm \delta_2, \bm C} -n\ln|\bm C|+\sum_{i \in G_1}(\bm w_i-\frac{2n_2}{n}\bm \delta_2-\bar{\bm w})'\bm C(\bm w_i-\frac{2n_2}{n}\bm \delta_2-\bar{\bm w})\nonumber\\
&+\sum_{i \in G_2}(\bm w_i+\frac{2n_1}{n}\bm \delta_2-\bar{\bm w})'\bm C(\bm w_i+\frac{2n_1}{n}\bm \delta_2-\bar{\bm w})+\lambda_1||\bm C||_1+\lambda_2|\bm \delta_2|_1.
\end{align}
In this manner, solving the optimization problem \eqref{eq: 1} is equivalent to solving the optimization problem \eqref{eq: obj}.
Next, we show that \eqref{eq: obj} can be decomposed as a graphical lasso model (Glasso) \citep{yuan2007model, deng2009large} in terms of $\bm C$ and a Lasso regression \citep{tibshirani1996regression} in terms of $\bm \delta_2$ with the other parameter fixed,
such that these two parameters can be estimated iteratively.
To be more precise, for a given value of $\bm \delta_2$, the minimization problem (\ref{eq: obj}) with respect to $\bm C$ is
\begin{equation}\label{eq: 2}
\min_{\bm C} -n\ln|\bm C|+\hbox{tr}(\bm C\tilde{\bm S})+\lambda_1||\bm C||_1,
\end{equation}
where $\tilde{\bm S}=\sum_{i\in G_1}(\bm w_i-\frac{2n_2}{n}\bm \delta_2-\bar{\bm w})(\bm w_i-\frac{2n_2}{n}\bm \delta_2-\bar{\bm w})' + \sum_{i \in G_2}(\bm w_i+\frac{2n_1}{n}\bm \delta_2-\bar{\bm w})(\bm w_i+\frac{2n_1}{n}\bm \delta_2-\bar{\bm w})'$.
It has the same form as the graphical lasso, which has been extensively studied in literature by \cite{yuan2007model, friedman2008sparse, lam2009sparsistency, raskutti2009model, liu2020minimax}, and many others.
On the other hand, when the inverse covariance matrix $\bm C$ is fixed, the minimization problem (\ref{eq: obj}) regarding $\bm \delta_2$ becomes
\begin{align}\label{eq: obj2more}
\min_{\bm \delta_2}&\sum_{i\in G_1}(\bm w_i-\frac{2n_2}{n}\bm \delta_2-\bar{\bm w})'\bm C(\bm w_i-\frac{2n_2}{n}\bm \delta_2-\bar{\bm w})  \nonumber \\
&+\sum_{i \in G_2}(\bm w_i+\frac{2n_1}{n}\bm \delta_2-\bar{\bm w})'\bm C(\bm w_i+\frac{2n_1}{n}\bm \delta_2-\bar{\bm w})+\lambda_2|\bm \delta_2|_1,
\end{align}
which is equivalent to
\begin{equation}\label{eq: objstep3}
\min_{\bm \delta_2} (\tilde{\bm y} - \bm C^{1/2} \bm \delta_2)'(\tilde{\bm y} - \bm C^{1/2} \bm \delta_2)+\lambda_2|\bm \delta_2|_1,
\end{equation}
where $\tilde{\bm y}= \frac{1}{2n_1n_2}\bm C^{1/2}(n_2\sum\limits_{i \in G_1}\bm w_i - n_1\sum\limits_{i \in G_2}\bm w_i)$.
A detailed derivation of (\ref{eq: objstep3}) from \eqref{eq: obj2more} is provided in the Appendix.
We solve the minimization problem (\ref{eq: objstep3}) by the Lasso technique.
Consequently, solving the complicated optimization problem \eqref{eq: obj} is decomposed to the simple tasks of iteratively solving a Glasso estimate for $\bm C$ and a Lasso estimate for $\bm \delta_2$ until both of them are converged.
We summarize the above estimation procedure for the proposed model in Algorithm \ref{alg1}.

\begin{algorithm}[Estimation Procedure]\label{alg1}
\noindent

{\bf Step 0}: Set an initial value of $\bm \delta_2$.

{\bf Step 1}: Given $\bm \delta_2=\hat{\bm \delta}_{2,t}$, solve $\bm C$ in (\ref{eq: 2}) by the Glasso technique.

{\bf Step 2}: Given $\bm C= \hat{\bm C}_t$, solve $\bm \delta_2$ in (\ref{eq: objstep3}) by the Lasso technique.

{\bf Step 3}: Repeat \emph{Step 1} and \emph{2} till both $\hat{\bm C}_t$ and $\hat{\bm \delta}_{2,t}$ converge.

\end{algorithm}

Here $\hat{\bm C}_t$ and $\hat{\bm \delta}_{2,t}$ represent the estimates of $\bm C$ and $\bm \delta_2$ in the $t$th iteration.
The convergence criteria are $||\hat{\bm C}_t-\hat{\bm C}_{t-1}||_F^2<\tau_1$ and $||\hat{\bm \delta}_{2,t}-\hat{\bm \delta}_{2,{t-1}}||_2^2<\tau_2$,
where $\tau_1$ and $\tau_2$ are two pre-selected small quantities, $||\cdot||_F$ stands for the Frobenius norm,
and $\| \bm \alpha \|_2^2 = \mathop{\sum}\limits_i \alpha_i^2$ with $\alpha_i$ being the $i$th entry of vector $\bm \alpha$.
We set the initial value of $\bm \delta_2$ as $(\bar{\bm w}_1-\bar{\bm w}_2)/2$,
where $\bar{\bm w}_k$ is the sample mean for the $k$th class.
With value of $\hat{\bm \delta}_2$, the estimate $\hat{\bm \gamma}$ is calculated by Equation (\ref{eq: gamma}),
and then we have $\hat{\bm \mu}_1=\hat{\bm \delta}_2+\hat{\bm \gamma}$ and $\hat{\bm \mu}_2=\hat{\bm \gamma}-\hat{\bm \delta}_2$.
Therefore, Algorithm \ref{alg1} provides the estimates of three parameters $\bm \mu_1$, $\bm \mu_2$ and $\bm C$ in the classification rule \eqref{eq: rule}.


Note that there are two tuning parameters $\lambda_1$ and $\lambda_2$ in the optimization problem \eqref{eq: obj}.
To choose their optimal values, we minimize a BIC-type criterion proposed by \cite{wang2007tuning} as
\begin{align*}
\mbox{BIC}(\lambda_1,\lambda_2)=-n\ln|\hat{\bm C}| + \tr(\hat{\bm C}\tilde{\bm S})& + (v(\hat{\bm \delta}_2) + v(\hat{\bm C}) + 1) \ln(n),
\end{align*}
where $v(\hat{\bm \delta}_2)$ and $v(\hat{\bm C})$ stand for the number of nonzero entries in the estimates $\hat{\bm \delta}_2$ and $\hat{\bm C}$, respectively.
This criterion enjoys consistency properties and has been commonly used in literature \citep{zou2009on, lv2009a, armagan2013generalized}.

\subsection{Model Prediction}
In this section, we demonstrate how to conduct model prediction by the proposed method.
For convenience, let us write
\begin{align*}
\bm \mu_1 = \left[
\begin{array}{cc}
\bm \mu_{1X} \\
\mu_{1y}
\end{array}\right],~~~
\bm \mu_2 = \left[
\begin{array}{cc}
\bm \mu_{2X} \\
\mu_{2y}
\end{array}\right],~~~ \mbox{and} ~~~
\bm C = \left[
\begin{array}{cc}
\bm C_{X},      & \bm C_{Xy} \\
\bm C_{Xy}',   & c_{y}^2
\end{array}\right],
\bm \Sigma= \left[
\begin{array}{cc}
\bm \Sigma_{X},      & \bm \Sigma_{Xy} \\
\bm \Sigma_{Xy}',   & \sigma_{y}^2
\end{array}\right].
\end{align*}
where $\bm \mu_{1X}$ and $\bm \mu_{2X}$ are $p-1$ dimensional vectors representing the means of variable $\bm X$ for two classes, and $\bm \Sigma_{X}$ is the $(p-1) \times (p-1)$ covariance matrix of $\bm X$.
The estimates $\hat{\bm \mu}_1, \hat{\bm \mu}_2, \hat{\bm C}, \hat{\bm \Sigma}$ can be partitioned accordingly.
From model assumption \eqref{assumption} as well as the property of multivariate normal distribution, we have
$y | \bm X = \bm x, Z = 1 \sim N \left( \mu_{1y} + \bm \Sigma_{Xy}' \bm \Sigma_{X}^{-1} (\bm x - \bm \mu_{1X}), \sigma_{y}^2 - \bm \Sigma_{Xy}' \bm \Sigma_{X}^{-1} \bm \Sigma_{Xy} \right)$, and
$y | \bm X = \bm x, Z = 2 \sim N \left( \mu_{2y} + \bm \Sigma_{Xy}' \bm \Sigma_{X}^{-1} (\bm x - \bm \mu_{2X}), \sigma_{y}^2 - \bm \Sigma_{Xy}' \bm \Sigma_{X}^{-1} \bm \Sigma_{Xy} \right)$.
Therefore, the prediction for the quantitative variable $y$ from a new observation $\bm x$ is
\begin{align}\label{ypred}
\hat{y} = \left\{
\begin{array}{cccc}
\hat{\mu}_{1y} + \hat{\bm \Sigma}_{Xy}' \hat{\bm \Sigma}_{X}^{-1} (\bm x - \hat{\bm \mu}_{1X}), \mbox{ if }  \hat{Z} = 1 \\
\hat{\mu}_{2y} + \hat{\bm \Sigma}_{Xy}' \hat{\bm \Sigma}_{X}^{-1} (\bm x - \hat{\bm \mu}_{2X}), \mbox{ if }  \hat{Z} = 2.
\end{array}
\right.
\end{align}
Note that $\hat{\bm \Sigma}_{Xy}' \hat{\bm \Sigma}_{X}^{-1} = -\frac{1}{\hat{c}_{y}^2} \hat{\bm C}_{Xy}'$ where $\hat{c}_{y}^2$ is a scalar, implying that the sparsity of $\hat{\bm C}_{Xy}$ will lead to the sparse model for the prediction of $y$.

On the other hand, the prediction for the qualitative variable $Z$ by the proposed model is naturally based on the estimated LDA classification rule of \eqref{eq: rule} as
\begin{equation}\label{estrule}
\ln\frac{\hbox{Pr}(G_1|\bm W=\bm (\bm x', \hat{y})')}{\hbox{Pr}(G_2|\bm W=\bm (\bm x', \hat{y})')}=\ln\frac{\hat{\pi}_1}{\hat{\pi}_2}-\frac{1}{2}(\hat{\bm \mu}_1+\hat{\bm \mu}_2)' \hat{\bm C} \hat{\bm \delta} + \bm (\bm x', \hat{y}) \hat{\bm C} \hat{\bm \delta}.
\end{equation}
From Equations \eqref{ypred} and \eqref{estrule}, however, we note that the prediction of one response variable depends on the information of the other.
To address this issue, we propose to calculate two candidate values of $y$ for a new observation $\bm x$ by Equation \eqref{ypred} for two different classes, denoted by $\hat{y}_1$ and $\hat{y}_2$.
Then the conditional probability densities $p(\bm W = (\bm x', \hat{y}_1)' | G_1 )$ and $p(\bm W = (\bm x', \hat{y}_2)' | G_2 )$ can be estimated via the density functions of $N(\hat{\bm \mu}_1, \hat{\bm \Sigma})$ and $N(\hat{\bm \mu}_2, \hat{\bm \Sigma})$.
Denote such two values as $\hat{p}_1$ and $\hat{p}_2$.
The prediction of $y$ at this new observation is then obtained as $\hat{y}_k$ corresponding to the larger value of $\hat{\pi}_{k} \hat{p}_k, k = 1, 2$.
To express it clearly, we describe the above steps of the model prediction in Algorithm \ref{alg2} for a new observation $\bm x$.

%
%
%
%
\begin{algorithm}[Prediction Procedure]\label{alg2}
\noindent

{\bf Step 1}: For $k = 1, 2$, $\hat{y}_k = \hat{\mu}_{ky} + \hat{\bm \Sigma}_{Xy}' \hat{\bm \Sigma}_{X}^{-1} (\bm x - \hat{\bm \mu}_{kX})$,
and consequently obtain the probability densities $\hat{p}_k$ by plugging $(\bm x', \hat{y}_k)'$ into the density functions of $N(\hat{\bm \mu}_k, \hat{\bm \Sigma})$.

{\bf Step 2a}: If $\hat{\pi}_1 \hat{p}_1 > \hat{\pi}_2 \hat{p}_2$, let $\hat{y} = \hat{y}_{1}$; otherwise let $\hat{y} = \hat{y}_{2}$.

{\bf Step 2b}: Apply the LDA classification rule \eqref{estrule} to predict $Z$ by $\bm w = (\bm x', \hat{y})'$.
\end{algorithm}

It is seen that in Algorithm \ref{alg2}, we obtain the prediction of $y$ first, and then predict $Z$ with the value of $\hat{y}$.
One would argue that it is not a unique way of making predictions on QQ responses, as we may also predict $Z$ first and then variable $y$.
The following proposition provides an interesting insight into this issue.

\begin{proposition}{}{\label{proposition1}}
For the prediction of variable $Z$ by the proposed model, the class label $k$ obtained from {\bf Step 2b} of Algorithm \ref{alg2} maximizes $\hat{\pi}_k \hat{p}_k$.
\end{proposition}

Proposition \ref{proposition1} implies that we can predict the response variable $Z$ by simply comparing values of $\hat{\pi}_k \hat{p}_k$ instead of employing LDA.
Therefore, the order of which response variable to be predicted first is not a concern.
Actually, the {\bf Step 2a} and {\bf Step 2b} are equivalent to the following {\bf Step 2} as

{\bf Step 2}: If $\hat{\pi}_1 \hat{p}_1 > \hat{\pi}_2 \hat{p}_2$, let $\hat{y} = \hat{y}_{1}$ and $\hat{Z} = 1$; otherwise let $\hat{y} = \hat{y}_{2}$ and $\hat{Z} = 2$.

\subsection{Extension to Multi-Class Qualitative Response}
The proposed generative modeling approach also has the advantage to enable the GAQQ to deal with the qualitative response with multiple classes,
i.e., the qualitative variable $Z \in \{1,2, \ldots, K\}$. In such cases, the GAQQ method is extended and expressed as $G_{k}: \bm W|Z = k \sim N(\bm \mu_k, \bm \Sigma), k = 1, 2, \ldots, K$.
Based on a baseline class $G_{1}$, we regularize on the difference between means through $\bm \mu_{k} - \bm \mu_{1}$ for $k = 2, 3, \ldots, K$.
The objective function is thus formulated as
\begin{align}\label{eq: obj-multi-class1}
\min_{(\bm \mu_1, \ldots, \bm \mu_K,\bm C)} -n\ln|\bm C| + \sum_{k=1}^{K} \sum_{i\in G_k}(\bm w_i-\bm \mu_k)' \bm C (\bm w_i-\bm \mu_k)+\lambda_1||\bm C||_1 + \lambda_2 \sum_{k=2}^{K}|\bm \mu_k-\bm \mu_1|_1.
\end{align}

The subsequent derivation follows similar steps as that described in Section \ref{subsec:derivation}.
With a little abuse of notation, let $K \bm \delta_k = \bm \mu_k - \bm \mu_1$ for $k = 1, 2, \ldots, K$ and $K \bm \gamma = \sum_{k=1}^{K} \bm \mu_k$, then we have $\bm \mu_{k} = \bm \gamma - \sum_{g=2}^{K} \bm \delta_g + K \bm \delta_k, ~ k = 1, 2, \ldots, K$.
As a result, the optimization problem \eqref{eq: obj-multi-class1} can be re-written as
\begin{align}\label{eq: obj-multi-class2}
\min_{(\bm \delta_2, \ldots, \bm \delta_K, \bm \gamma, \bm C)} -n \ln|\bm C| + \sum_{k=1}^{K} \sum_{i\in G_k} &(\bm w_i - \bm \gamma + \sum_{g=2}^{K} \bm \delta_g - K \bm \delta_k)' \bm C (\bm w_i - \bm \gamma + \sum_{g=2}^{K} \bm \delta_g - K \bm \delta_k) \nonumber \\
&+ \lambda_1||\bm C||_1 + \lambda_2 \sum_{k=2}^{K} | \bm \delta_k |_1.
\end{align}
Let $n_k$ represent the number of observations belonging to class $G_k$.
The maximum likelihood estimator of $\bm \gamma$ from \eqref{eq: obj-multi-class2} is $\hat{\bm \gamma} = \bar{\bm w} + \sum_{g=2}^{K} \bm \delta_g - \frac{K}{n} \sum_{g=2}^{K} n_g \bm \delta_g$.
Consequently, the optimization problem \eqref{eq: obj-multi-class2} becomes
\begin{align}\label{eq: obj-multi-class3}
\min_{(\bm \delta_2, \ldots, \bm \delta_K, \bm C)} -n \ln|\bm C| + \sum_{k=1}^{K} \sum_{i\in G_k} &(\bm w_i - \bar{\bm w} + \frac{K}{n} \sum_{g=2}^{K} n_g \bm \delta_g - K \bm \delta_k)' \bm C (\bm w_i - \bar{\bm w} + \frac{K}{n} \sum_{g=2}^{K} n_g \bm \delta_g        \nonumber \\
& - K \bm \delta_k) + \lambda_1||\bm C||_1 + \lambda_2 \sum_{k=2}^{K} | \bm \delta_k |_1.
\end{align}
Let $\tilde{\bm S} = \sum_{k=1}^{K} \sum_{i\in G_k} (\bm w_i - \bar{\bm w} + \frac{K}{n} \sum_{g=2}^{K} n_g \bm \delta_g - K \bm \delta_k)(\bm w_i - \bar{\bm w} + \frac{K}{n} \sum_{g=2}^{K} n_g \bm \delta_g - K \bm \delta_k)'$, then the formula \eqref{eq: obj-multi-class3} can be decomposed as one Glasso problem
\begin{align*}
\min_{\bm C} -n \ln|\bm C| + \tr(\bm C \tilde{\bm S}) + \lambda_1||\bm C||_1,
\end{align*}
and
\begin{align}\label{eq: obj-multi-class-lasso1}
\min_{(\bm \delta_2, \ldots, \bm \delta_K)} \sum_{k=1}^{K} \sum_{i\in G_k} (\bm w_i - \bar{\bm w} + \frac{K}{n} \sum_{g=2}^{K} n_g \bm \delta_g - K \bm \delta_k)' \bm C (\bm w_i - \bar{\bm w} &+ \frac{K}{n} \sum_{g=2}^{K} n_g \bm \delta_g
- K \bm \delta_k) \nonumber \\
&+ \lambda_2  \sum_{k=2}^{K} | \bm \delta_k |_1.
\end{align}
The optimization problem \eqref{eq: obj-multi-class-lasso1} is equivalent to the following $K - 1$ Lasso regressions separately
\begin{align}\label{eq: obj-multi-class-lasso2}
\min_{\bm \delta_k} (\tilde{\bm y} - \bm C^{1/2} \bm \delta_k)'(\tilde{\bm y} - \bm C^{1/2} \bm \delta_k) + \lambda_2 |\bm \delta_k|_1, ~~~ k = 2, 3, \ldots, K,
\end{align}
where $\tilde{\bm y}= \frac{1}{K n n_k} \bm C^{1/2} \left[(n - n_k) \sum\limits_{i \in G_k}\bm w_i - n_k \sum\limits_{i \notin G_k}\bm w_i + K n_k \sum\limits_{g = 2, g \neq k}^K n_g \bm \delta_g \right]$.
The detailed derivation from \eqref{eq: obj-multi-class-lasso1} to \eqref{eq: obj-multi-class-lasso2} is provided in the Appendix.
Therefore, the parameters $\bm \delta_k$ and $\bm C$ can be solved iteratively until convergence following the spirit of Algorithm \ref{alg1}.
The optimal values of tuning parameters are chosen by the BIC-type criterion extended for the multi-class problem as
\begin{align*}
\mbox{BIC}(\lambda_1,\lambda_2)=-n\ln|\hat{\bm C}| + \tr(\hat{\bm C}\tilde{\bm S})& + (v(\hat{\bm \delta}) + v(\hat{\bm C}) + K - 1) \ln(n),
\end{align*}
where $v(\hat{\bm \delta})$ represents the number of nonzero entries in all the estimates $\hat{\bm \delta}_k$.

For a new observation $\bm x$, the quantitative response $y$ is predicted, similarly as in Algorithm \ref{alg2}, to be $\hat{y}_k = \hat{\mu}_{ky} + \hat{\bm \Sigma}_{Xy}' \hat{\bm \Sigma}_{X}^{-1} (\bm x - \hat{\bm \mu}_{kX})$, where $k$ maximizes $\hat{\pi}_k \hat{p}_k$ with $\hat{p}_k=p(\bm W = (\bm x', \hat{y}_k)' | G_k )$, computed by plugging $(\bm x', \hat{y}_k)'$ into the density functions of $N(\hat{\bm \mu}_k, \hat{\bm \Sigma})$.
The class label is estimated
as $\hat{Z} = \mathop{\arg \max}\limits_{k} \hat{\pi}_k \hat{p}_k$, or equivalently by the LDA rule as
\begin{align*}
\hat{Z} = \arg\max_{k} ~ \ln \frac{\hat{\pi}_k}{\hat{\pi}_1} + K ((\bm x', \hat{y}_k)' - \frac{\hat{\bm \mu}_1 + \hat{\bm \mu}_k}{2})' \hat{\bm C} \hat{\bm \delta}_k.
\end{align*}

\section{Theoretical Properties}\label{sec:theory}
In this section, we will investigate the asymptotic optimality of the classification rule by the proposed GAQQ method in Theorem \ref{theory:multiclass} to Theorem \ref{theory:supp}.
The asymptotic consistency properties of the prediction of $y$ by the GAQQ method are established in Theorem \ref{theory:y}.
For the proposed classification rule, we first establish the theoretical results for the multi-class problem and then provide a thorough discussion of the two-class case.
We use the same definition of asymptotic optimality for a classification rule as defined in \cite{Shao2011Sparse}.
Denote by $R_{Bayes}$ and $R_{PROP}(\mathcal{T})$ the Bayes error and the conditional misclassification rate of the proposed rule, where $\mathcal{T}$ denotes the training samples.
The asymptotic optimality for a classification rule is defined as follows.
\begin{definition}{}
\noindent Let $T$ be a classification rule with conditional misclassification rate $R_T(\mathcal{T})$,
given the training samples $\mathcal{T}$. \\
(1) $T$ is asymptotically optimal if $R_T(\mathcal{T}) / R_{Bayes} \stackrel{P}{\rightarrow} 1$. \\
(2) $T$ is asymptotically sub-optimal if $R_T(\mathcal{T}) - R_{Bayes} \stackrel{P}{\rightarrow} 0$.
\end{definition}
Note that if $\lim\limits_{n \rightarrow \infty} R_{Bayes} > 0$,
then the asymptotically sub-optimality is the same as the asymptotically optimality.
To facilitate the construction of theoretical results, we need to introduce some notation and make assumptions on the true model.
Define the true values of $\bm \mu_k$, $\bm \Sigma$, $\bm C$ and $\bm \delta_k$ as $\bm \mu_k^0$, $\bm \Sigma^0$, $\bm C^0$ and $\bm \delta_k^0 = \frac{1}{K} (\bm \mu_k^0 - \bm \mu_1^0) = ((\bm \delta_{kX}^0)', \delta_{ky}^0)'$, where $\bm \delta_{kX}^0$ is a $p-1$  dimensional vector representing the true mean difference of variable $\bm X$ between classes $G_1$ and $G_k$.
Denote the true inverse covariance matrix of variable $\bm X$ by $\bm C_X^0$.
Also define $\Delta_k = \sqrt{(\bm \delta_{kX}^0)' \bm C_X^0 \bm \delta_{kX}^0}$ and $\Delta = \max \{ \Delta_k \}_{k=1}^K$.
Denote $\mathbb{S}_{\delta_k} = \{j; (\bm \delta_k^0)_j \neq 0 \}$, which is the set containing location indices of the nonzero entries in $\bm \delta_k^0$.
Let $\tilde{s}_k$ be the cardinalities of set $\mathbb{S}_{\delta_k}$.
Define $s_k = \tilde{s}_k$ if $\delta_{ky}^0 = 0$; otherwise $s_k = \tilde{s}_k - 1$.
That is, $s_k$ is the number of nonzero entries of $\bm \delta_{kX}^0$.
Additionally, we use the same sparsity measure on $\bm \Sigma^0 = (\sigma_{ij}^0)_{1 \leq i, j \leq p}$ as in \cite{Bickel2008Covariance}, which is $S_{h;p} = \max_{i \leq p} \sum_{j=1}^{p} |\sigma_{ij}^0|^h$
where $0 \leq h < 1$ and $0^0$ is defined to be 0.
Hence firstly, $S_{0;p}$ equals the maximum of the numbers of nonzero entries in each row of the matrix $\bm \Sigma^0$.
In this case, a smaller value of $S_{0;p}$ compared with $p$ implies a sparse structure in matrix $\bm \Sigma^0$.
Secondly, if $S_{h;p}$ is smaller than $p$ for $0 < h < 1$, it indicates that many entries of matrix $\bm \Sigma^0$ are very small.
Moreover, we assume the following regularity conditions.
\begin{itemize}
\item(C1)  There exists a constant $\theta$ such that $0 < \theta^{-1} < \lambda_{\min}(\bm C^0) \leq \lambda_{\max}(\bm C^0) < \theta < \infty$,
where $\lambda_{\min}(\bm C^0)$ and $\lambda_{\max}(\bm C^0)$ are the minimum and maximum
eigenvalues of matrix $\bm C^0$.

\item(C2) $\lambda_1 = O(\sqrt{\log p / n})$, $\lambda_2 = O(\sqrt{\log p / n})$.

\item(C3) Restricted eigenvalue condition: for some constant $\varphi_k > 0$, assume $\bm C^0$ satisfies
$\frac{1}{n}\| (\bm C^0)^{1/2} \bm \delta_k^0 \|_2^2 \geq \varphi_k \| \bm \delta_k^0 \|_2^2$ for all subsets $J \subseteq \{1,\ldots,p\}$ such that the cardinality of $J$ equals $\tilde{s}_k$, and $| (\bm \delta_k^0)_{J^c} |_1 \leq 3| (\bm \delta_k^0)_{J} |_1$.
Here $(\bm \delta_k^0)_{J} = ((\bm \delta_k^0)_j \cdot I\{j \in J\})_{1\leq j \leq p}$, and $J^c$ represents the complement set of $J$.

\item(C4) Irrepresentable condition: without loss of generality,
write $\bm \delta_k^0 = ((\bm \delta_k^0)'_{\mathbb{S}_{\delta}},
(\bm \delta_k^0)'_{\mathbb{S}_{\delta}^c})'$,
and correspondingly let $\bm C^0 = \left[
\begin{array}{cc}
\bm \Psi_{11}, & \bm\Psi_{12} \\
\bm \Psi_{21}, & \bm\Psi_{22}
\end{array}\right]$,
where $\bm \Psi_{11}$ is an $\tilde{s}_k \times \tilde{s}_k$ matrix.
Then there exists a positive constant vector $\bm \zeta$ such that
$| \bm \Psi_{21} \bm \Psi_{11}^{-1} \mbox{sign}((\bm \delta_k^0)'_{\mathbb{S}_{\delta}}) | \leq \bm 1 - \bm \zeta$,
where $\bm 1$ is a $p - \tilde{s}_k$ dimensional unit vector,
and the inequality holds element-wise.

\item(C5) There exist $0 \leq c_1 < c_2 \leq 1$ and $M > 0$, such that
$n^{\frac{1-c_2}{2}} \min_{1 \leq i \leq \tilde{s}_k}| (\bm \delta_k^0)_i| \geq M$, $\tilde{s}_k = O(n^{c_1})$. $\lambda_2 = o(n^{\frac{c_2-c_1+1}{2}})$,
$p = o(\lambda_2^2 / n)$.

\item(C6)  There exists a constant $c_3 > 0$ such that $(\bm \delta_{kX}^0 - \bm \delta_{lX}^0)' \bm C_X^0 (\bm \delta_{kX}^0 - \bm \delta_{lX}^0) > c_3 > 0$ for $k \neq l$.

\item(C7)  There exists a constant $c_4$ such that
$c_4^{-1} \leq K \pi_k \leq c_4, k = 1, 2, \ldots, K$.
\end{itemize}

By conditions (C1) and (C2), \cite{rothman2008sparse} and \cite{lam2009sparsistency}
derived the convergence rate of Glasso estimate.
We thus have
\begin{equation}\label{conrateC}
\| \hat{\bm C}_X - \bm C_X^0 \| = O_p(d_n),
\end{equation}
where $d_n = S_{h;p} (\frac{\log p}{n})^{(1-h)/2}$, and $\|\bm A\|$ is the matrix spectral norm
defined as the squared root of the maximum eigenvalue of matrix $\bm A' \bm A$.
The conditions (C2) and (C3) are used in \cite{B2011Statistics} to study the theoretical property of Lasso estimate, and we have
\begin{equation}\label{conratedelta}
\| \hat{\bm \delta}_{kX} - \bm \delta_{kX}^0 \|_2 = O_p(b_k^{(n)}),
\end{equation}
where $b_k^{(n)}= \sqrt{\frac{\tilde{s}_k \log p}{n \varphi_k^2}}$.
Under conditions (C4) and (C5), \cite{zhao2006on} showed that the Lasso estimate is model selection consistency,
which will be used for investigating $\bm \delta_{k}$ in \eqref{eq: obj-multi-class-lasso2}.
Condition (C6) requires that all the classes should be separated from each other. 
Also note that condition (C6) is equivalent to that $\Delta_k$ is bounded away from 0.
The condition (C7) guarantees a balanced sample size for each class, which is commonly used in literature to bound the term $\log \frac{\pi_k}{\pi_1}$ in the LDA rule for establishing the properties of classification rules.
Based on the above results, we present the following theories on the consistency of the classification rule by the proposed method.

\begin{theorem}{}{\label{theory:multiclass}}
Assume that conditions (C1) - (C7) hold, and
\begin{align*}
\xi_{n;k} = \max \{ d_n, \frac{b_k^{(n)}}{\Delta_k}, \frac{\sqrt{s_k S_{h;p}}}{\sqrt{n} \Delta_k} ~\mbox{for any}~k \} \rightarrow 0.
\end{align*}
Then the proposed rule for the multi-class problem is asymptotically sub-optimal if either one of the following two conditions is satisfied  \\
(1) $\Delta = \max \{ \Delta_k \}_{k=1}^K$ is bounded; \\
(2) if $\Delta \rightarrow \infty$, then there exists a constant $\alpha \in (0, 1/2)$  such that $\Delta^2 \xi_{n;k}^{1-2 \alpha} \rightarrow 0$.
\end{theorem}

Theorem \ref{theory:multiclass} establishes the sub-optimality property of the proposed classification rule for the multi-class problem.
In the case of two-class problem, the Bayes error can be expressed in a closed form of $R_{Bayes} = \Phi( - \Delta_2 / 2)$
when the data are from normal distribution, where $\Phi$ represents the cumulative distribution function of $N(0, 1)$,
and $\Delta_2 = \sqrt{(\bm \delta_{2X}^0)' \bm C_X^0 \bm \delta_{2X}^0} = \sqrt{(\bm \mu_{2X}^0 - \bm \mu_{1X}^0)' \bm C_X^0 (\bm \mu_{2X}^0 - \bm \mu_{1X}^0)}$.
Accordingly, in Theorems \ref{theory:consistent} and \ref{theory:supp}, we can compute the convergence rate of the proposed rule for the two-class problem, and subsequently investigate its properties.

\begin{theorem}{}{\label{theory:consistent}}
Assume that conditions (C1) - (C7) hold with $K = 2$, and
\begin{align*}
\xi_n = \max \{ d_n, \frac{b_2^{(n)}}{\Delta_2}, \frac{\sqrt{s_2 S_{h;p}}}{\sqrt{n} \Delta_2} \} \rightarrow 0.
\end{align*}
Then we have $R_{PROP}(\mathcal{T}) = \Phi(-\frac{\Delta_2}{2} [1 + O_p(\xi_n)])$.
\end{theorem}
Moreover, we establish the following properties.
\begin{theorem}{}{\label{theory:supp}}
Assuming all the conditions in Theorem \ref{theory:consistent} are satisfied, we have

\noindent (1) if $\Delta_2$ is bounded, then the proposed rule is asymptotically optimal and
$\frac{R_{PROP}(\mathcal{T})}{R_{Bayes}} - 1 = O_p(\xi_n)$; \\
(2) if $\Delta_2 \rightarrow \infty$, then the proposed rule is asymptotically sub-optimal;  \\
(3) if $\Delta_2 \rightarrow \infty$ and $\xi_n \Delta_2^2 \rightarrow 0$,
then the proposed rule is asymptotically optimal.
\end{theorem}

Theorem \ref{theory:consistent} provides the convergence rate of the proposed classification rule for the two-class problem with respect to $\xi_n$.
Base on such a result, Theorem \ref{theory:supp} demonstrates that the property of the proposed classification rule (optimality or sub-optimality) depends on the scenarios of the true model's $\Delta_2$.
Specifically, (1) when $\Delta_2$ is bounded, i.e. $\lim\limits_{n \rightarrow \infty} R_{Bayes} > 0$, then $R_{PROP}(\mathcal{T})$ converges in probability to the same limit as $R_{Bayes}$.
(2) When $\Delta_2 \rightarrow \infty$, i.e. $R_{Bayes} \rightarrow 0$, then   $R_{PROP}(\mathcal{T}) \stackrel{P}{\rightarrow} 0$;
in this case, if we further have $\xi_n \Delta_2^2 \rightarrow 0$, then $R_{PROP}(\mathcal{T})$ and $R_{Bayes}$ have the same convergence rate.

Next, we derive the consistency property for the proposed estimate of $y$.
Denote by $\hat{y}^{P}$ the predicted value of $y$ obtained from the proposed model.
Define $\hat{y}^{B}$ to be the predicted value of $y$ for $\bm x$ when all parameters are known.
Specifically, first obtain the class label $k$ via the Bayes LDA rule, then $\hat{y}^{B} = y_k = \mu_{ky} + \bm \Sigma_{Xy}' \bm \Sigma_{X}^{-1} (\bm x - \bm \mu_{kX})$.
Hence, the mean squared errors ($MSE$) of estimates $\hat{y}^{B}$ and $\hat{y}^{P}$ are $MSE_{Bayes} = \E[(\hat{y}^{B}-y)^2|\mathcal{T}]$ and $MSE_{PROP} = \E[(\hat{y}^{P}-y)^2|\mathcal{T}]$.
Now we establish the theoretical results of $\hat{y}^{P}$ in Theorem \ref{theory:y}.

\begin{theorem}{}{\label{theory:y}}
Assume that conditions (C1) - (C7) hold and conditions in Theorem \ref{theory:multiclass} are satisfied.
Then we have
\begin{align*}
MSE_{PROP} - MSE_{Bayes} \stackrel{P}{\rightarrow} 0,
\end{align*}
for the multi-class qualitative response.
\end{theorem}
This result compares the $MSE$ of the proposed estimate of $y$ with that from the optimal Bayes rule (under all parameters known).
Since the classification errors from a classification rule might be larger than 0,
the $MSE$ of $\hat{y}$ may not converge to 0 even though the sample size $n$ is sufficiently large.
Here we adopt the $MSE_{Bayes}$ as a reasonable performance benchmark to evaluate the property of the proposed model with respect to $y$.
Theorem \ref{theory:y} states that the difference of $MSE$ between the proposed model and the Bayes method converges to 0 in probability.

\section{Simulation}\label{sec:simulation}
\subsection{Two-class Settings of the Qualitative Response}

In this section, we evaluate the performance of the proposed GAQQ method for a binary response $Z$ under different inverse covariance matrices $\bm C$ and mean differences $\bm \delta_2$.
The proposed GAQQ model is compared with several benchmark methods, denoted as GLDA, CL, and ENET, which use the predictor variables $\bm X$ to predict $Z$ and $y$.
The GLDA employs the LDA classification rule for $Z$ using the generalized inverse of the sample covariance matrix of $\bm X$ when $p > n$.
The CL method applies the LPD technique introduced by \cite{cai2012a} to predict the response variable $Z$ based on $\bm X$.
With their estimated class label of $Z$, the GLDA and CL predict $y$ by Equation \eqref{ypred}.
The ENET method uses the elastic-net logistic model \citep{zou2005regularization} on predictor variables $\bm X$ to fit the qualitative response $Z$ and hence predicts $Z$ for the testing data.
For the quantitative response $y$, the ENET separately fits two elastic-net linear regressions for two classes using training data and then predicts $y$ in the testing data based on its estimated $Z$.
The tuning parameters of the CL and ENET methods are chosen by cross-validation.

Regarding the inverse covariance matrix $\bm C$, we consider the following five structures in the simulation, which are commonly used in the literature \citep{yuan2007model, kang2020variable}.
\begin{itemize}
\item Model 1. $\bm C_1$ = $\bm I$. $c_{ij}=1$ if $i=j$ and 0 otherwise;

\item Model 2. $\bm C_2$ = AR(0.6). The conditional covariance between any two random variables is fixed to be $0.6^{|i - j|}$, $1 \leq i, j \leq p$.

\item Model 3. $\bm C_3$ is generated by randomly permuting rows and corresponding columns of the matrix $\bm C_2$.

\item Model 4. $\bm C_4 = \left(
\begin{array}{cc}
    \mbox{CS}(0.6) & \bm 0 \\
    \bm 0 &  \bm I
  \end{array}
\right)$, where CS(0.6) represents a $5 \times 5$ compound symmetry matrix with diagonal entries 1 and others 0.6.
$\bm 0$ indicates a matrix with all entries 0.

\item Model 5. $\bm C_{5} = \bm \Theta + \alpha \bm I$, where the diagonal entries of $\bm \Theta$ are zeros and $\bm \Theta_{ij} = \bm \Theta_{ji} = b * Unif(-1, 1)$ for $i \neq j$, where $b$ is from the Bernoulli distribution with probability 0.15 equal 1. Each off-diagonal entry of $\bm \Theta$ is generated independently.
The value of $\alpha$ is gradually increased to make sure that $\bm C_{5}$ is positive definite.
\end{itemize}
Model 1 is the simplest sparse matrix indicating that variables are independent of each other.
Model 4 is a sparse matrix indicating that only the first 5 variables are correlated.
This matrix has more sparsity as the dimensionality increases.
Models 2 and 3 are relatively dense matrices, and they also become more sparse when the dimensionality increases.
All of these four matrices have sparse structures to some extent,
while Model 5 is a general sparse matrix with no structure, which is similarly used in \cite{bien2011sparse}.

For the mean difference $\bm \delta_2$, we consider two different levels of sparsity.
The $\bm \mu_1$ is the vector with all elements zeros.
Then generate $\bm\mu_2$ such that (S1): $25\%$ of the elements in $\bm\mu_2$ are zeros; (S2): $75\%$ of the elements in $\bm\mu_2$ are zeros.
The positions of zeros in $\bm \mu_{2}$ are randomly distributed with its
nonzero values independently generated from uniform distribution $Unif(0, 2)$.
We consider $p \in \{40, 80, 200\}$, and generate $n_1 = 30$ observations from $N(\bm \mu_{1}, \bm C^{-1})$ as well as
$n_2 = 30$ observations from $N(\bm \mu_{2}, \bm C^{-1})$ as the training set.
The same procedure is employed to generate the testing data, which is used to evaluate the prediction performance of $y$ and $Z$ for different compared methods.
We consider the root mean squared prediction error $\textrm{RMSPE} = \sqrt{\frac{1}{n} \sum_{i=1}^{n} (y_{i} - \hat{y}_{i})^2}$ to measure the prediction accuracy for the quantitative response $y$, where $\hat{y}_{i}$ represents the predicted value.
The prediction performance of the qualitative response $Z$ is measured by the misclassification error $\textrm{ME} = \frac{1}{n}\sum_{i=1}^n I(z_{i} \neq \hat{z}_{i})$, where $\hat{z}_{i}$ is the predicted value of $z_{i}$ and $I(\cdot)$ is an indicator function.

\begin{table}
\begin{center}
\caption{Averages and standard errors (in parenthesis) of misclassification errors (MEs) in percentage for methods in comparison.}
\label{table:me}
\begin{tabular}{cccrrrrrrrrrrr}
\hline \hline
\multicolumn{3}{c}{}&Model 1&Model 2&Model 3&Model 4&Model 5\\
\hline
\multirow{8}{*}{$p=40$}
&  &ENET&25.3(0.51) &25.4(0.44) &25.0(0.49) &24.8(0.46) &25.6(0.45) \\
&S1&GLDA&9.28(0.42) &10.2(0.75) &9.40(0.66) &5.28(0.33) &9.13(0.33) \\
&  &CL  &1.68(0.21) &10.9(1.04) &8.75(1.00) &4.58(0.51) &1.20(0.22) \\
&  &GAQQ&2.92(0.26) &9.93(0.41) &17.7(0.50) &2.33(0.24) &2.67(0.22) \\
\hhline{~---------}
&  &ENET&25.3(0.45) &26.1(0.47) &24.3(0.42) &26.2(0.42) &25.1(0.41) \\
&S2&GLDA&20.0(0.66) &9.78(0.41) &13.1(0.49) &22.8(0.63) &18.8(0.58)\\
&  &CL  &6.92(0.39) &4.90(0.33) &8.63(0.50) &8.65(0.38) &8.48(0.40) \\
&  &GAQQ&5.32(0.32) &4.52(0.28) &7.10(0.28) &6.28(0.30) &5.88(0.29) \\
\hline
\multirow{8}{*}{$p=80$}
&  &ENET&25.0(0.51) &24.4(0.44) &24.8(0.44) &23.9(0.48) &25.8(0.54) \\
&S1&GLDA&7.88(0.43) &11.2(0.49) &12.6(0.52) &8.03(0.38) &11.3(0.49) \\
&  &CL  &6.37(1.38) &10.4(0.66) &11.2(0.63) &5.63(1.48) &9.38(1.57) \\
&  &GAQQ&0.10(0.04) &2.97(0.22) &2.22(0.19) &0.07(0.03) &3.53(0.22) \\
\hhline{~---------}
&  &ENET&24.8(0.42) &25.3(0.39) &25.2(0.43) &24.3(0.43) &23.9(0.52) \\
&S2&GLDA&19.5(0.55) &26.7(0.68) &24.8(0.76) &16.1(0.63) &20.9(0.61) \\
&  &CL  &4.67(0.37) &20.5(0.77) &18.2(1.01) &2.65(0.30) &14.5(0.92) \\
&  &GAQQ&1.08(0.13) &10.6(0.37) &5.42(0.30) &0.57(0.11) &3.33(0.22) \\
\hline
\multirow{8}{*}{$p=200$}
&  &ENET&24.3(0.41) &24.9(0.52) &24.4(0.42) &25.2(0.43) &24.6(0.43) \\
&S1&GLDA&2.25(0.20) &14.1(0.52) &13.6(0.42) &3.13(0.22) &7.23(0.36) \\
&  &CL  &2.08(0.20) &2.88(0.18) &2.72(0.17) &2.12(0.21) &2.26(0.15) \\
&  &GAQQ&0.22(0.06) &0.47(0.09) &0.20(0.05) &0.23(0.07) &0.05(0.03) \\
\hhline{~---------}
&  &ENET&25.3(0.40) &25.5(0.40) &24.6(0.47) &25.7(0.49) &25.3(0.51) \\
&S2&GLDA&9.73(0.40) &20.5(0.55) &24.3(0.57) &9.08(0.40) &15.0(0.50) \\
&  &CL  &1.46(0.14) &2.96(0.16) &2.29(0.11) &2.06(0.18) &2.38(0.16) \\
&  &GAQQ&0.01(0.00) &1.10(0.16) &1.55(0.16) &0.02(0.01) &0.17(0.05) \\
\hline \hline
\end{tabular}
\end{center}
\end{table}

Tables \ref{table:me} and \ref{table:rmse} report the averaged MEs in percentage and averaged RMPSE, as well as their corresponding standard errors in parenthesis for each approach over 100 replications.
It can be seen from Table \ref{table:me} that the proposed method generally outperforms other approaches with respect to MEs.
Such an advantage becomes more significant as the underlying models are more sparse.
Specifically, in the scenario of S1 = $25\%$ and $p = 40$, the proposed GAQQ method does not perform as well as others, since the underlying models in this scenario are the least sparse, especially for the dense models 2 and 3.
In contrast, the proposed method produces relatively better comparison results in the scenario of S2 = $75\%$ and $p = 40$, where the true mean difference is more sparse.
Furthermore, this advantage of the proposed method is well evidenced in the scenario of $p = 80$, and even more notable when $p = 200$ with its substantially lower MEs than other methods.

\begin{table}
\begin{center}
\caption{Averages and standard errors (in parenthesis) of root mean squared prediction errors (RMSPE) for methods in comparison.}
\label{table:rmse}
\begin{tabular}{cccrrrrrrrrrrr}
\hline \hline
\multicolumn{3}{c}{}&Model 1&Model 2&Model 3&Model 4&Model 5\\
\hline
\multirow{8}{*}{$p=40$}
&  &ENET&1.18(0.01) &1.62(0.02) &1.84(0.02) &1.91(0.01) &1.73(0.02) \\
&S1&GLDA&1.82(0.03) &2.00(0.04) &2.03(0.04) &1.82(0.03) &1.82(0.03) \\
&  &CL  &1.79(0.03) &1.97(0.06) &2.03(0.08) &1.65(0.03) &1.74(0.03) \\
&  &GAQQ&1.07(0.01) &1.21(0.01) &1.49(0.01) &1.02(0.01) &1.17(0.02) \\
\hhline{~---------}
&  &ENET&1.59(0.01) &1.20(0.01) &1.42(0.02) &1.22(0.01) &1.12(0.01) \\
&S2&GLDA&1.93(0.03) &1.96(0.03) &1.90(0.03) &1.85(0.03) &1.77(0.02)\\
&  &CL  &1.82(0.03) &1.78(0.03) &1.70(0.03) &1.67(0.03) &1.58(0.03) \\
&  &GAQQ&1.07(0.01) &1.14(0.01) &1.37(0.01) &0.98(0.01) &1.09(0.01) \\
\hline
\multirow{8}{*}{$p=80$}
&  &ENET&1.58(0.01) &1.77(0.02) &1.68(0.02) &1.37(0.01) &1.75(0.01) \\
&S1&GLDA&2.01(0.03) &2.62(0.04) &2.38(0.04) &2.08(0.03) &1.92(0.03) \\
&  &CL  &2.02(0.04) &2.63(0.07) &2.31(0.05) &1.88(0.03) &1.72(0.03) \\
&  &GAQQ&1.11(0.01) &1.26(0.01) &1.54(0.01) &1.11(0.01) &1.31(0.01) \\
\hhline{~---------}
&  &ENET&1.08(0.01) &1.31(0.01) &1.44(0.02) &1.61(0.01) &1.11(0.01) \\
&S2&GLDA&1.96(0.03) &2.56(0.05) &2.27(0.04) &2.04(0.03) &2.36(0.04) \\
&  &CL  &1.76(0.03) &2.38(0.05) &2.10(0.04) &1.85(0.03) &2.20(0.04) \\
&  &GAQQ&1.02(0.01) &1.10(0.01) &1.39(0.01) &0.99(0.01) &1.11(0.01) \\
\hline
\multirow{8}{*}{$p=200$}
&  &ENET&1.66(0.02) &1.24(0.01) &1.68(0.02) &1.07(0.01) &1.61(0.02) \\
&S1&GLDA&1.24(0.01) &1.58(0.02) &1.62(0.02) &1.21(0.01) &1.40(0.02) \\
&  &CL  &1.27(0.03) &1.60(0.02) &1.65(0.02) &1.28(0.02) &1.36(0.03) \\
&  &GAQQ&1.08(0.01) &1.27(0.01) &1.44(0.02) &1.06(0.01) &1.15(0.01) \\
\hhline{~---------}
&  &ENET&1.03(0.01) &1.65(0.01) &1.62(0.01) &1.22(0.01) &1.17(0.01) \\
&S2&GLDA&1.19(0.01) &1.67(0.02) &1.76(0.02) &1.22(0.01) &1.32(0.01) \\
&  &CL  &1.19(0.01) &1.55(0.02) &1.74(0.02) &1.23(0.01) &1.33(0.01) \\
&  &GAQQ&1.01(0.01) &1.25(0.01) &1.43(0.01) &1.01(0.01) &1.15(0.01) \\
\hline \hline
\end{tabular}
\end{center}
\end{table}

From Table \ref{table:rmse}, we observe that the proposed method generally gives superior performance over other compared approaches for each scenario in predicting the quantitative response $y$.
The possible explanations are in two folds.
First, the proposed GAQQ method provides an accurate classification of the qualitative response $Z$.
Second, the proposed GAQQ has a proper estimation of $\bm C$ by the regularization that is used in the prediction of quantitative response $y$ according to \eqref{ypred}, resulting in an improvement of the prediction accuracy.
It is also seen that the CL and GLDA methods are comparable in some cases,
possibly because both of them use the generalized inverse of the sample covariance of $\bm X$ for $\hat{\bm \Sigma}_{X}^{-1}$ in the prediction of quantitative response $y$ in \eqref{ypred}.
But the CL method is generally better since it has more accurate classification results than the GLDA in Table \ref{table:me}.


\subsection{Multi-class Settings of the Qualitative Response}

Now, we examine the performance of the proposed GAQQ method for multi-class settings of the qualitative response.
We consider $p = 200$ and $K = 4$ classes of qualitative response $Z$
with training sizes $n_1 = n_2 = n_3 = n_4 = 30$ for Models 1 - 5 of inverse covariance matrix $\bm C$.
Let $\mu_{kj}$ represent the $j$th entry of the mean value $\bm \mu_k$.
Generate $\mu_{kj} = 0.5 * k + u_{kj}$ for $j = 2k - 1, 2k, 2k + 1, \ldots, 2k + 6$, otherwise $\mu_{kj} = 0$, where $u_{kj}$ is from $Unif(-1, 1)$.
The training data are generated from $N(\bm \mu_k, \bm C^{-1})$,
and the testing data follow the same generation procedure.
We compare the proposed method with the GLDA, as well as the estimators proposed by  \cite{Witten2011Penalized} (WT) and \cite{Clemmensen2011Sparse} (CHWE),
where the latter two methods are designed for multi-class problems.
We use the WT and CHWE models to first predict the class label $Z$ for the testing data, and then the response $y$ is estimated, by the multivariate normal property, as $\hat{\mu}_{ky} + \hat{\bm \Sigma}_{Xy}' \hat{\bm \Sigma}_{X}^{-1} (\bm x - \hat{\bm \mu}_{kX})$ if their estimates $\hat{Z} = k$.
The results of performance measures, ME and RMSPE are summarized in Table \ref{table:multi} based on 100 replications.
One can see that the GAQQ method performs better than the GLDA as well as the WT method,
and is comparable with the CHWE in terms of the MEs.
Besides, the GAQQ method gives the best performance among the compared approaches with significantly lower values of RMSPE.

\begin{table}[h]
\small
\begin{center}
\caption{Averages and standard errors (in parenthesis) of MEs in percentage and RMSPE for methods in comparison for multi-class settings of $p = 200$.}
\label{table:multi}
\begin{tabular}{rrrrrrrrrrrrrrrr}
\hline \hline
&&Model 1 &Model 2  &Model 3  &Model 4  &Model 5   \\
\hline
&&&& ME    \\ \hline
&GLDA   &40.58 (0.46) &59.15 (0.56) &55.56 (0.51) &39.40 (0.48) &48.28 (0.55) \\
&WT     &17.26 (0.35) &43.73 (0.59) &43.01 (0.47) &17.90 (0.38) &33.17 (0.67)  \\
&CHWE   &14.70 (0.32) &25.14 (0.40) &32.16 (0.50) &16.66 (0.44) &21.31 (0.44)  \\
&GAQQ   &14.02 (0.32) &25.36 (0.53) &33.34 (0.50) &17.11 (0.42) &20.99 (0.49)  \\
\hline
&&&& RMSPE     \\
\hline
&GLDA   &1.64 (0.01)  &2.09 (0.02)  &2.02 (0.02)  &1.66 (0.01) &1.71 (0.02) \\
&WT     &1.56 (0.01)  &2.05 (0.02)  &1.94 (0.02)  &1.57 (0.01) &1.63 (0.02)  \\
&CHWE   &1.56 (0.01)  &2.01 (0.02)  &1.92 (0.02)  &1.55 (0.01) &1.61 (0.02)  \\
&GAQQ   &0.99 (0.01)  &1.11 (0.01)  &1.27 (0.01)  &1.01 (0.01) &1.39 (0.02)  \\
\hline \hline
\end{tabular}
\end{center}
\end{table}

\section{Case Studies}\label{sec:realdata}
%
%

In this section, we apply the proposed GAQQ method to two real-data case studies.
The first one is from the study of Heusler compounds in material science and the second one is from the study of molecular diagnostics of Ulcerative colitis and Crohn's disease.
Although from different fields, both problems contain QQ responses with high-dimensional predictors, and the proposed GAQQ method appears to have much better performance in terms of prediction accuracy compared with other methods.

The case study on material sciences is regarding the Heusler compounds, which are a large family of intermetallics with more than 1000 known members.
Many Heusler compounds have shown exotic properties, such as superconductivity and topological band structures, which have promising applications for quantum computing.
Understanding the thermodynamic stability of Heusler compounds lays the foundation for exploiting the large chemical space to discover and design new functional Heusler materials \citep{liu2016observation}.
To determine the thermodynamic stability of Heusler compounds, there are two key metrics: the mixing enthalpy (quantitative response) and the global stability based on hull energy (binary qualitative response).
The comprehensive database of 180628 full Heusler structures was built by collecting the relevant structural and energetic data from the Materials Project \citep{jain2013commentary}, OQMD \citep{saal2013materials}, and AFLOW \citep{curtarolo2012aflow}.
These data were calculated using first-principles methods based on density functional theory, and it was extremely computationally expensive (taking hours) to generate one entry of the data.
Therefore, a statistical model that can accurately predict the thermodynamic stability for any elemental and compound features is a useful surrogate of the first-principle computation models.

Since there is an intrinsic relationship between two QQ responses, the proposed GAQQ method is suitable to improve the prediction accuracy by jointly fitting them together.
To demonstrate the GAQQ method in the scenario when the number of predictors is large relative to the size of the data, we randomly choose 150 samples from each class of the binary response.
We delete the predictor variables whose standard deviations are less than $1.0e^{-6}$, resulting in 157 predictors of elemental and compound features.
To examine the prediction performance of the GAQQ method and other comparison methods,
we randomly divide data into a training set with a size of $200$ and a testing set with a size of $100$.
Table \ref{table:realdata} reports the prediction performance results based on 50 random splits of the Heusler data.
From the results,
it is seen that the proposed GAQQ performs much better than other methods in comparison, with the smallest values for the misclassification error (ME)
and the root mean squared prediction error (RMSPE).

\begin{table}[h]
\small
\begin{center}
\caption{The MEs in percentage and RMSPE of Heusler and gene expression data.} \label{table:realdata}
\begin{tabular}{rrrrrrrrrrr}
\hline \hline
&&& Heusler Data \\
\hline
&Methods    &GLDA           &ENET               &CL             &GAQQ    \\
&ME         &27.27 (1.828)  &11.87 (0.332)      &16.20 (0.688)  &10.49 (0.363) \\
&RMSPE      &1.797 (0.445)  &0.317 (0.083)      &1.046 (0.053)  &0.142 (0.002) \\
\hline
&&& IBD Gene Data \\
\hline
&Methods    &GLDA           &WT               &CHWE           &GAQQ    \\
&ME         &21.90 (0.800)  &24.80 (0.583)    &18.10 (0.555)  &15.77 (0.584) \\
&RMSPE      &0.743 (0.014)  &0.751 (0.014)    &0.746 (0.014)  &0.661 (0.011) \\
\hline \hline
\end{tabular}
\end{center}
\end{table}

The second data for the case study considers the multi-class settings of the qualitative response.
The IBD gene data \citep{Burczynski2006Molecular} are gene expressions on Ulcerative colitis (UC) and Crohn's disease (CD), two of which are common inflammatory bowel diseases (IBD) producing intestinal inflammation and tissue damage.
The IBD data set was collected at North American and European clinical sites from blood samples of 42 healthy individuals, 59 CD patients, and 26 UC patients with 22,283 genes.
An exploratory analysis, similarly conducted as in \cite{Shao2011Sparse}, is performed as variable screening by one-way ANOVA with three levels (healthy individuals, CD patients, and UC patients).
We choose the top 101 significant gene variables to form the data for methods comparison.
To create a quantitative response, one gene variable is randomly chosen as the quantitative response from the 101 significant variables.
The data set is then randomly partitioned into a training set with 67 samples and testing data with the rest 60 samples.
Table \ref{table:realdata} presents the comparison results by the GLDA, WT, CHWE, and proposed GAQQ methods based on 50 random splits of the data.
We observe that the proposed GAQQ method performs substantially well with relatively lower values of ME and RMSPE, as well as their corresponding standard errors in the parenthesis.
Such empirical results demonstrate that the proposed GAQQ method can achieve accurate predictions for both QQ responses in high-dimensional data.

\section{Discussion}\label{sec:discussion}
In this work, we propose a generative modeling approach to jointly model the data with QQ responses, which is a new perspective different from existing methods in the literature.
By fully exploring the joint distribution of the QQ responses and predictor variables, the proposed method enables efficient parameter estimation, accurate prediction,
and lays a good foundation for investigating the asymptotic properties.
The proposed model can be naturally extended to the situation for multiple quantitative responses.

One further research direction is to accommodate a more flexible structure on the joint distribution of QQ responses and predictor variables.
For example, one can extend the LDA for the classification of the qualitative response to the quadratic discriminant analysis (QDA).
The QDA is more flexible with different covariance structures in each class,
but its estimation for high-dimensional data would encounter more difficulty due to a large number of parameters.
Besides, the derivation of its asymptotic properties is much more technically complicated \citep{li2015sparse}.
Another research direction is to apply the generative modeling approach for the data with semi-continuous responses \citep{wang2020joint}, or the ordinal and quantitative responses.
One may employ the ordinal regression for the ordinal response, and then derive its joint likelihood function with appropriate regularization.

\bibliographystyle{vancouver-authoryear}
\bibliography{Ref_GAQQ}

\begin{thebibliography}{48}
\providecommand{\natexlab}[1]{#1}
\providecommand{\url}[1]{\texttt{#1}}
\providecommand{\urlprefix}{}

\bibitem[{Klein et~al.(2019)Klein, Nadja and Kneib, Thomas and Marra, Giampiero
  and Radice, Rosalba and Rokicki, Slawa and McGovern, Mark E}]{klein2019mixed}
Klein N, Kneib T, Marra G, Radice R, Rokicki S, McGovern ME.
\newblock Mixed binary-continuous copula regression models with application to
  adverse birth outcomes.
\newblock Statistics in Medicine 2019;38(3):413--436.

\bibitem[{Fitzmaurice and Laird(1995)Fitzmaurice, Garrett M and Laird, Nan
  M}]{fitzmaurice1995regression}
Fitzmaurice GM, Laird NM.
\newblock Regression models for a bivariate discrete and continuous outcome
  with clustering.
\newblock Journal of the American statistical Association
  1995;90(431):845--852.

\bibitem[{Moustaki and Knott(2000)Moustaki, Irini and Knott,
  Martin}]{moustaki2000generalized}
Moustaki I, Knott M.
\newblock Generalized latent trait models.
\newblock Psychometrika 2000;65(3):391--411.

\bibitem[{Dunson(2000)Dunson, David B}]{dunson2000bayesian}
Dunson DB.
\newblock Bayesian latent variable models for clustered mixed outcomes.
\newblock Journal of the Royal Statistical Society: Series B (Statistical
  Methodology) 2000;62(2):355--366.

\bibitem[{Gueorguieva and Agresti(2001)Gueorguieva, Ralitza V and Agresti,
  Alan}]{gueorguieva2001correlated}
Gueorguieva RV, Agresti A.
\newblock A correlated probit model for joint modeling of clustered binary and
  continuous responses.
\newblock Journal of the American Statistical Association
  2001;96(455):1102--1112.

\bibitem[{Dunson(2003)Dunson, David B}]{dunson2003dynamic}
Dunson DB.
\newblock Dynamic latent trait models for multidimensional longitudinal data.
\newblock Journal of the American Statistical Association
  2003;98(463):555--563.

\bibitem[{Deng and Jin(2015)Deng, Xinwei and Jin, Ran}]{deng2015qq}
Deng X, Jin R.
\newblock {QQ} models: Joint modeling for quantitative and qualitative quality
  responses in manufacturing systems.
\newblock Technometrics 2015;57(3):320--331.

\bibitem[{K{\"u}r{\"u}m et~al.(2016)K{\"u}r{\"u}m, Esra and Li, Runze and
  Shiffman, Saul and Yao, Weixin}]{kurum2016time}
K{\"u}r{\"u}m E, Li R, Shiffman S, Yao W.
\newblock Time-varying coefficient models for joint modeling binary and
  continuous outcomes in longitudinal data.
\newblock Statistica Sinica 2016;26(3):979--1000.

\bibitem[{Kang et~al.(2018)Kang, Lulu and Kang, Xiaoning and Deng, Xinwei and
  Jin, Ran}]{kang2018bayesian}
Kang L, Kang X, Deng X, Jin R.
\newblock A Bayesian hierarchical model for quantitative and qualitative
  responses.
\newblock Journal of Quality Technology 2018;50(3):290--308.

\bibitem[{Amini et~al.(2018)Amini, Payam and Verbeke, Geert and Zayeri, Farid
  and Mahjub, Hossein and Maroufizadeh, Saman and Moghimbeigi,
  Abbas}]{amini2018longitudinal}
Amini P, Verbeke G, Zayeri F, Mahjub H, Maroufizadeh S, Moghimbeigi A.
\newblock Longitudinal joint modelling of binary and continuous outcomes: A
  comparison of bridge and normal distributions.
\newblock Epidemiology, Biostatistics and Public Health 2018;15(1).

\bibitem[{Fitzmaurice and Laird(1997)Fitzmaurice, Garrett M and Laird, Nan
  M}]{fitzmaurice1997regression}
Fitzmaurice GM, Laird NM.
\newblock Regression models for mixed discrete and continuous responses with
  potentially missing values.
\newblock Biometrics 1997;53(1):110--122.

\bibitem[{Song et~al.(2009)Song, Peter X-K and Li, Mingyao and Yuan,
  Ying}]{song2009joint}
Song PXK, Li M, Yuan Y.
\newblock Joint regression analysis of correlated data using Gaussian copulas.
\newblock Biometrics 2009;65(1):60--68.

\bibitem[{Lin et~al.(2010)Lin, Lanjia and Bandyopadhyay, Dipankar and Lipsitz,
  Stuart R and Sinha, Debajyoti}]{lin2010association}
Lin L, Bandyopadhyay D, Lipsitz SR, Sinha D.
\newblock Association models for clustered data with binary and continuous
  responses.
\newblock Biometrics 2010;66(1):287--293.

\bibitem[{Chen et~al.(2014)Chen, Shizhe and Witten, Daniela M and Shojaie,
  Ali}]{chen2014selection}
Chen S, Witten DM, Shojaie A.
\newblock Selection and estimation for mixed graphical models.
\newblock Biometrika 2014;102(1):47--64.

\bibitem[{Yang et~al.(2014)Yang, E. and Baker, Y. and Ravikumar, P. and Allen,
  G. and Liu, Z.}]{yang2014mixed}
Yang E, Baker Y, Ravikumar P, Allen G, Liu Z.
\newblock Mixed graphical models via exponential families.
\newblock Proceedings of the Seventeenth International Conference on Artificial
  Intelligence and Statistics 2014;33:1042--1050.

\bibitem[{Guglielmi et~al.(2018)Alessandra Guglielmi and Francesca Ieva and
  Anna Maria Paganoni and Fernardo A. Quintana}]{guglielmi2018a}
Guglielmi A, Ieva F, Paganoni AM, Quintana FA.
\newblock A semiparametric Bayesian joint model for multiple mixed-type
  outcomes: an application to acute myocardial infarction.
\newblock Advances in Data Analysis and Classification 2018;12(2):399--423.

\bibitem[{Sammel et~al.(1997)Sammel, M. D. and Ryan, L. M. and Legler, J.
  M.}]{sammel1997latent}
Sammel MD, Ryan LM, Legler JM.
\newblock Latent variable models for mixed eiscrete and continuous outcomes.
\newblock Journal of the Royal Statistical Society: Series B (Statistical
  Methodology) 1997;59(3):667--678.

\bibitem[{Dunson and Herring(2005)Dunson, David B. and Herring, Amy
  H.}]{dunson2005bayesian}
Dunson DB, Herring AH.
\newblock Bayesian latent variable models for mixed discrete outcomes.
\newblock Biostatistics 2005;6(1):11--25.

\bibitem[{Bello et~al.(2012)Bello, N. M. and Steibel, J. P. and Tempelman, R.
  J.}]{bello2012hierarchical}
Bello NM, Steibel JP, Tempelman RJ.
\newblock Hierarchical Bayesian modeling of heterogeneous clusterand
  subject-level associations between continuous and binary outcomes in dairy
  production.
\newblock Biometrical Journal 2012;54(2):230--248.

\bibitem[{Shao et~al.(2011)Shao, Jun and Wang, Yazhen and Deng, Xinwei and
  Wang, Sijian}]{Shao2011Sparse}
Shao J, Wang Y, Deng X, Wang S.
\newblock Sparse linear discriminant analysis by thresholding for high
  dimensional data.
\newblock Annals of Statistics 2011;39(2):1241--1265.

\bibitem[{Zhao and Yu(2006)Zhao, Peng and Yu, Bin}]{zhao2006on}
Zhao P, Yu B.
\newblock On model selection consistency of Lasso.
\newblock Journal of Machine Learning Research 2006;7(12):2541--2563.

\bibitem[{Cai and Liu(2012)Cai, T. and Liu, W.}]{cai2012a}
Cai T, Liu W.
\newblock A direct estimation approach to sparse linear discriminant analysis.
\newblock Journal of the American Statistical Association 2012;106:1566--1577.

\bibitem[{Yuan and Lin(2007)Yuan, Ming and Lin, Yi}]{yuan2007model}
Yuan M, Lin Y.
\newblock Model selection and estimation in the Gaussian graphical model.
\newblock Biometrika 2007;94(1):19--35.

\bibitem[{Deng and Yuan(2009)Deng, Xinwei and Yuan, Ming}]{deng2009large}
Deng X, Yuan M.
\newblock Large Gaussian covariance matrix estimation with Markov structures.
\newblock Journal of Computational and Graphical Statistics
  2009;18(3):640--657.

\bibitem[{Tibshirani(1996)Tibshirani, Robert}]{tibshirani1996regression}
Tibshirani R.
\newblock Regression shrinkage and selection via the lasso.
\newblock Journal of the Royal Statistical Society: Series B (Statistical
  Methodology) 1996;58(1):267--288.

\bibitem[{Friedman et~al.(2008)Friedman, Jerome and Hastie, Trevor and
  Tibshirani, Robert}]{friedman2008sparse}
Friedman J, Hastie T, Tibshirani R.
\newblock Sparse inverse covariance estimation with the graphical lasso.
\newblock Biostatistics 2008;9(3):432--441.

\bibitem[{Lam and Fan(2009)Lam, Clifford and Fan,
  Jianqing}]{lam2009sparsistency}
Lam C, Fan J.
\newblock Sparsistency and rates of convergence in large covariance matrix
  estimation.
\newblock Annals of Statistics 2009;37(6B):4254--4278.

\bibitem[{Raskutti et~al.(2008)Raskutti, Garvesh and Yu, Bin and Wainwright,
  Martin J and Ravikumar, Pradeep}]{raskutti2009model}
Raskutti G, Yu B, Wainwright MJ, Ravikumar P.
\newblock Model Selection in Gaussian Graphical Models: High-Dimensional
  Consistency of l1-regularized MLE.
\newblock Advances in Neural Information Processing Systems 2008;21:1329--1336.

\bibitem[{Liu et~al.(2020)Liu, Yu and Ren, Zhao and others}]{liu2020minimax}
Liu Y, Ren Z, et~al.
\newblock Minimax estimation of large precision matrices with bandable Cholesky
  factor.
\newblock Annals of Statistics 2020;48(4):2428--2454.

\bibitem[{Wang et~al.(2007)Wang, H. and Li, R. and Tsai, C.L.}]{wang2007tuning}
Wang H, Li R, Tsai CL.
\newblock Tuning parameter selectors for the smoothly clipped absolute
  deviation method.
\newblock Biometrika 2007;94(3):553--568.

\bibitem[{Zou and Zhang(2009)Zou, H. and Zhang, H.}]{zou2009on}
Zou H, Zhang H.
\newblock On the adaptive elastic-net with a diverging number of parameters.
\newblock Annals of Statistics 2009;37(4):1733--1751.

\bibitem[{Lv and Fan(2009)Lv, Jinchi and Fan, Yingying}]{lv2009a}
Lv J, Fan Y.
\newblock A unified approach to model selection and sparse recovery using
  regularized least squares.
\newblock Annals of Statistics 2009;37(6A):3498--3528.

\bibitem[{Armagan et~al.(2013)Armagan, A. and Dunson, D. B. and Lee,
  J.}]{armagan2013generalized}
Armagan A, Dunson DB, Lee J.
\newblock Generalized double pareto shrinkage.
\newblock Statistica Sinica 2013;23(1):119--143.

\bibitem[{Bickel and Levina(2008)Bickel, Peter J. and Levina,
  Elizaveta}]{Bickel2008Covariance}
Bickel PJ, Levina E.
\newblock Covariance regularization by thresholding.
\newblock Annals of Statistics 2008;36(6):2577--2604.

\bibitem[{Rothman et~al.(2008)Rothman, Adam J and Bickel, Peter J and Levina,
  Elizaveta and Zhu, Ji and others}]{rothman2008sparse}
Rothman AJ, Bickel PJ, Levina E, Zhu J, et~al.
\newblock Sparse permutation invariant covariance estimation.
\newblock Electronic Journal of Statistics 2008;2:494--515.

\bibitem[{B\"{u}hlmann and Van De~Geer(2011)B\"{u}hlmann, Peter and Van De
  Geer, Sara}]{B2011Statistics}
B\"{u}hlmann P, Van De~Geer S.
\newblock Statistics for High-Dimensional Data.
\newblock Verlag Berlin Heidelberg: Springer; 2011.

\bibitem[{Zou and Hastie(2005)Zou, Hui and Hastie,
  Trevor}]{zou2005regularization}
Zou H, Hastie T.
\newblock Regularization and variable selection via the elastic net.
\newblock Journal of the Royal Statistical Society: Series B (Statistical
  Methodology) 2005;67(2):301--320.

\bibitem[{Kang and Deng(2020)Kang, Xiaoning and Deng,
  Xinwei}]{kang2020variable}
Kang X, Deng X.
\newblock On variable ordination of Cholesky-based estimation for a sparse
  covariance matrix.
\newblock Canadian Journal of Statistics 2020;in press.

\bibitem[{Bien and Tibshirani(2011)Bien, J. and Tibshirani, R.
  J.}]{bien2011sparse}
Bien J, Tibshirani RJ.
\newblock Sparse estimation of a covariance matrix.
\newblock Biometrika 2011;98(4):807--820.

\bibitem[{Witten and Tibshirani(2011)Witten, D. M. and Tibshirani,
  R.}]{Witten2011Penalized}
Witten DM, Tibshirani R.
\newblock Penalized classification using Fisher's linear discriminant.
\newblock Journal of the Royal Statistical Society: Series B (Statistical
  Methodology) 2011;73(5):753--772.

\bibitem[{Clemmensen et~al.(2011)Clemmensen, Line and Hastie, Trevor and
  Witten, Daniela and Ersb${\o}$ll, Bjarne}]{Clemmensen2011Sparse}
Clemmensen L, Hastie T, Witten D, Ersb${\o}$ll B.
\newblock Sparse Discriminant Analysis.
\newblock Technometrics 2011;53(4):406--413.

\bibitem[{Liu et~al.(2016)Liu, ZK and Yang, LX and Wu, S-C and Shekhar, Chandra
  and Jiang, Juan and Yang, HF and Zhang, Yi and Mo, S-K and Hussain, Zahid and
  Yan, Binghai and others}]{liu2016observation}
Liu Z, Yang L, Wu SC, Shekhar C, Jiang J, Yang H, et~al.
\newblock Observation of unusual topological surface states in half-Heusler
  compounds LnPtBi (Ln=Lu,Y).
\newblock Nature Communications 2016;7(1):1--7.

\bibitem[{Jain et~al.(2013)Jain, Anubhav and Ong, Shyue Ping and Hautier,
  Geoffroy and Chen, Wei and Richards, William Davidson and Dacek, Stephen and
  Cholia, Shreyas and Gunter, Dan and Skinner, David and Ceder, Gerbrand and
  others}]{jain2013commentary}
Jain A, Ong SP, Hautier G, Chen W, Richards WD, Dacek S, et~al.
\newblock Commentary: The Materials Project: A materials genome approach to
  accelerating materials innovation.
\newblock Apl Materials 2013;1(1):011002.

\bibitem[{Saal et~al.(2013)Saal, James E and Kirklin, Scott and Aykol,
  Muratahan and Meredig, Bryce and Wolverton, Christopher}]{saal2013materials}
Saal JE, Kirklin S, Aykol M, Meredig B, Wolverton C.
\newblock Materials design and discovery with high-throughput density
  functional theory: the open quantum materials database (OQMD).
\newblock Jom 2013;65(11):1501--1509.

\bibitem[{Curtarolo et~al.(2012)Curtarolo, Stefano and Setyawan, Wahyu and
  Hart, Gus LW and Jahnatek, Michal and Chepulskii, Roman V and Taylor, Richard
  H and Wang, Shidong and Xue, Junkai and Yang, Kesong and Levy, Ohad and
  others}]{curtarolo2012aflow}
Curtarolo S, Setyawan W, Hart GL, Jahnatek M, Chepulskii RV, Taylor RH, et~al.
\newblock AFLOW: an automatic framework for high-throughput materials
  discovery.
\newblock Computational Materials Science 2012;58:218--226.

\bibitem[{Burczynski et~al.(2006)Burczynski, Michael E and Peterson, Ron L and
  Twine, Natalie C and Zuberek, Krystyna A and Brodeur, Brendan J and
  Casciotti, Lori and Maganti, Vasu and Reddy, Padma S and Strahs, Andrew and
  Immermann, Fred and others}]{Burczynski2006Molecular}
Burczynski ME, Peterson RL, Twine NC, Zuberek KA, Brodeur BJ, Casciotti L,
  et~al.
\newblock Molecular classification of Crohn's disease and ulcerative colitis
  patients using transcriptional profiles in peripheral blood mononuclear
  cells.
\newblock The Journal of Molecular Diagnostics 2006;8(1):51--61.

\bibitem[{Li and Shao(2015)Li, Quefeng and Shao, Jun}]{li2015sparse}
Li Q, Shao J.
\newblock Sparse quadratic discriminant analysis for high dimensional data.
\newblock Statistica Sinica 2015;25:457--473.

\bibitem[{Wang et~al.(2020)Wang, Xiaoqing and Feng, Xiangnan and Song,
  Xinyuan}]{wang2020joint}
Wang X, Feng X, Song X.
\newblock Joint analysis of semicontinuous data with latent variables.
\newblock Computational Statistics and Data Analysis 2020;p. 107005.

\end{thebibliography}

\clearpage
\section*{Appendix}
\textit{Derivation from} \eqref{eq: obj2more} to \eqref{eq: objstep3}.
Let $C$ denote a generic constant thereafter.
\begin{align*}
&\sum\limits_{i\in G_1}(\bm w_i-\frac{2n_2}{n}\bm \delta_2-\bar{\bm w})'\bm C(\bm w_i-\frac{2n_2}{n}\bm \delta_2-\bar{\bm w})\\
&+\sum\limits_{i \in G_2}(\bm w_i+\frac{2n_1}{n}\bm \delta_2-\bar{\bm w})'\bm C(\bm w_i+\frac{2n_1}{n}\bm \delta_2-\bar{\bm w})+\lambda_2|\bm \delta_2|_1\\
=&\sum_{i\in G_1}(\bm C^{1/2}\bm w_i-\frac{2n_2}{n}\bm C^{1/2}\bm \delta_2-\bm C^{1/2}\bar{\bm w})'(\bm C^{1/2}\bm w_i-\frac{2n_2}{n}\bm C^{1/2}\bm \delta_2-\bm C^{1/2}\bar{\bm w})\\
&+\sum_{i\in G_2}(\bm C^{1/2}\bm w_i+\frac{2n_1}{n}\bm C^{1/2}\bm \delta_2-\bm C^{1/2}\bar{\bm w})'(\bm C^{1/2}\bm w_i+\frac{2n_1}{n}\bm C^{1/2}\bm \delta_2-\bm C^{1/2}\bar{\bm w})+\lambda_2|\bm \delta_2|_1\\
= &\sum\limits_{i\in G_1}(-2(\frac{2n_2}{n}\bm C^{1/2}\bm \delta_2)'(\bm C^{1/2}\bm w_i-\bm C^{1/2}\bar{\bm w})+\frac{4n_2^2}{n^2}\bm \delta_2'\bm C\bm \delta_2)\\
&+\sum\limits_{i\in G_2}(2(\frac{2n_1}{n}\bm C^{1/2}\bm \delta_2)'(\bm C^{1/2}\bm w_i-\bm C^{1/2}\bar{\bm x})+\frac{4n_1^2}{n^2}\bm \delta_2'\bm C\bm \delta_2)+\lambda_2|\bm \delta_2|_1 + C \\
=&-\frac{4n_2}{n}\bm \delta_2'\bm C\sum_{i\in G_1}\bm w_i+\frac{4n_2}{n}\bm \delta_2' \bm C(n_1\bar{\bm w})+\frac{4n_1n_2^2}{n^2}\bm \delta_2'\bm C\bm \delta_2\\
&+\frac{4n_1}{n}\bm \delta_2'\bm C\sum_{i\in G_2}\bm w_i-\frac{4n_1}{n}\bm \delta_2' \bm C(n_2\bar{\bm w})+\frac{4n_1^2n_2}{n^2}\bm \delta_2'\bm C\bm \delta_2+\lambda_2|\bm \delta_2|_1 + C \\
=&\frac{4n_1n_2}{n}\bm \delta_2'\bm C\bm \delta_2+\frac{4n_1}{n}\bm \delta_2'\bm C(n\bar{\bm w})-4\bm \delta_2'\bm C\sum_{i\in G_1}\bm w_i-\frac{4n_1}{n}\bm \delta_2'\bm C(n\bar{\bm w})        +4\bm \delta_2'\bm C(n_1\bar{\bm w}) + \lambda_2|\bm \delta_2|_1 + C   \\
=&\frac{4n_1n_2}{n}\bm \delta_2'\bm C\bm \delta_2-4\bm \delta_2'\bm C(\sum_{i \in G_1}\bm w_i-n_1\bar{\bm w})+\lambda_2|\bm \delta_2|_1 + C \\
= &\frac{4n_1n_2}{n}(\tilde{\bm y}-\bm C^{1/2}\bm \delta_2)'(\tilde{\bm y} - \bm C^{1/2} \bm \delta_2)+\lambda_2|\bm \delta_2|_1 + C,
\end{align*}
\vskip -5 pt
\noindent where $\tilde{\bm y}=\frac{n}{2n_1n_2}\bm C^{1/2}(\sum\limits_{i \in G_1}\bm w_i-n_1\bar{\bm w}) = \frac{1}{2n_1n_2}\bm C^{1/2}(n_2\sum\limits_{i \in G_1}\bm w_i - n_1\sum\limits_{i \in G_2}\bm w_i)$.
~~~~~~~~~~~~~~~~~~~~~~~~~~~~~~~~~~~~~~~$\Box$

\noindent \textit{Derivation from \eqref{eq: obj-multi-class-lasso1} to \eqref{eq: obj-multi-class-lasso2}.}
For $\bm \delta_j, j = 2,3,\ldots,K$,
\begin{align}\label{derivation1}
&\sum_{k=1}^{K} \sum_{i \in G_k} (\bm w_i - \bar{\bm w} + \frac{K}{n} \sum_{g=2}^{K} n_g \bm \delta_g - K \bm \delta_k)' \bm C (\bm w_i - \bar{\bm w} + \frac{K}{n} \sum_{g=2}^{K} n_g \bm \delta_g - K \bm \delta_k) + \lambda_2 | \bm \delta_j |_1  \nonumber \\
=& \sum_{k=1, k \neq j}^{K} \sum_{i \in G_k} \left[ \bm C^{1/2} (\bm w_i - \bar{\bm w} + \frac{K}{n} \sum_{g=2, g \neq j}^{K} n_g \bm \delta_g - K \bm \delta_k) + \bm C^{1/2} \frac{K}{n} n_j \bm \delta_j \right]' \nonumber \\
&~~~~~~~~~~~~~~~~~~~\left[ \bm C^{1/2} (\bm w_i - \bar{\bm w} + \frac{K}{n} \sum_{g=2, g \neq j}^{K} n_g \bm \delta_g - K \bm \delta_k) + \bm C^{1/2} \frac{K}{n} n_j \bm \delta_j \right] \nonumber \\
&~~~~~~ + \sum_{i \in G_j} \left[ \bm C^{1/2} (\bm w_i - \bar{\bm w} + \frac{K}{n} \sum_{g=2, g \neq j}^{K} n_g \bm \delta_g) + \bm C^{1/2} (\frac{K}{n} n_j \bm \delta_j  - K \bm \delta_j) \right]'  \nonumber \\
&~~~~~~~~~~~~~~~~~~~\left[ \bm C^{1/2} (\bm w_i - \bar{\bm w} + \frac{K}{n} \sum_{g=2, g \neq j}^{K} n_g \bm \delta_g) + \bm C^{1/2} (\frac{K}{n} n_j \bm \delta_j  - K \bm \delta_j) \right] + \lambda_2 | \bm \delta_j |_1 \nonumber \\
=& \sum_{k=1, k \neq j}^{K} \sum_{i \in G_k} \left[ \frac{2K n_j}{n} (\bm C^{1/2} \bm \delta_j)'(\bm C^{1/2} \bm w_i - \bm C^{1/2} \bar{\bm w} + \frac{K}{n} \bm C^{1/2} \sum_{g=2, g \neq j}^{K} n_g \bm \delta_g - K \bm C^{1/2} \bm \delta_k) + \frac{K^2 n_j^2}{n^2} \bm \delta'_j \bm C \bm \delta_j \right] \nonumber \\
&~~~~~~ + \sum_{i \in G_j} \left[ 2 K (\frac{n_j}{n}-1) (\bm C^{1/2} \bm \delta_j)'(\bm C^{1/2} \bm w_i - \bm C^{1/2} \bar{\bm w} + \frac{K}{n} \bm C^{1/2} \sum_{g=2, g \neq j}^{K} n_g \bm \delta_g) + K^2 (\frac{n_j}{n}-1)^2 \bm \delta'_j \bm C \bm \delta_j  \right]   \nonumber \\
&~~~~~~~~ + \lambda_2 | \bm \delta_j |_1 + C \nonumber \\
=& \frac{K n_j}{n} \sum_{k=1, k \neq j}^{K} \left[ 2 \bm \delta'_j \bm C \sum_{i \in G_k} \bm w_i - 2 n_k \bm \delta'_j \bm C \bar{\bm w} + \frac{2 n_k K}{n} \bm \delta'_j \bm C \sum_{g=2, g \neq j}^{K} n_g \bm \delta_g - 2 n_k K \bm \delta'_j \bm C \bm \delta_k + \frac{K n_j n_k}{n} \bm \delta'_j \bm C \bm \delta_j \right]  \nonumber \\
&~~~~~~ + K n_j (\frac{n_j}{n}-1) \left( 2 \bm \delta'_j \bm C (\frac{1}{n_j} \sum_{i \in G_j} \bm w_i) - 2 \bm \delta'_j \bm C \bar{\bm w} + \frac{2K}{n} \bm \delta'_j \bm C \sum_{g=2, g \neq j}^{K} n_g \bm \delta_g + K (\frac{n_j}{n}-1) \bm \delta'_j \bm C \bm \delta_j \right)  \nonumber  \\
&~~~~~~~~ + \lambda_2 | \bm \delta_j |_1 + C \nonumber \\
=& \frac{K^2 n_j (n - n_j)}{n} \bm \delta'_j \bm C \bm \delta_j - \frac{2 K}{n} \bm \delta'_j \bm C \{ \sum_{k=1, k \neq j}^{K} (-n_j \sum_{i \in G_k} \bm w_i + n_j n_k \bar{\bm w} - \frac{K n_j n_k }{n} \sum_{g=2, g \neq j}^{K} n_g \bm \delta_g + K n_j n_k \bm \delta_k) \nonumber \\
&- (n_j - n) \sum_{i \in G_j} \bm w_i + n_j (n_j - n) \bar{\bm w} - K n_j (\frac{n_j}{n}-1) \sum_{g=2, g \neq j}^{K} n_g \bm \delta_g \} + \lambda_2 | \bm \delta_j |_1 + C \nonumber \\
\triangleq & \frac{K^2 n_j (n - n_j)}{n} \bm \delta'_j \bm C \bm \delta_j - \frac{2 K}{n} \bm \delta'_j \bm C M + \lambda_2 | \bm \delta_j |_1 + C,
\end{align}
where
\begin{align*}
M =& -n_j \sum_{i \notin G_j} \bm w_i + n_j (n - n_j) \bar{\bm w} - \frac{K n_j (n - n_j)}{n} \sum_{g=2, g \neq j}^{K} n_g \bm \delta_g + K n_j \sum_{g=2, g \neq j}^{K} n_g \bm \delta_g     \\
& - (n_j - n) \sum_{i \in G_j} \bm w_i + n_j (n_j - n) \bar{\bm w} - K n_j (\frac{n_j}{n}-1) \sum_{g=2, g \neq j}^{K} n_g \bm \delta_g       \\
=& (n - n_j) \sum\limits_{i \in G_j} \bm w_i - n_j \sum\limits_{i \notin G_j} \bm w_i + K n_j \sum\limits_{g=2, g \neq j}^{K} n_g \bm \delta_g.
\end{align*}
Let $\tilde{y} = \frac{1}{K n_j (n - n_j)} \bm C^{1/2} M = \frac{1}{K n_j (n - n_j)} \bm C^{1/2} \left[ (n - n_j) \sum\limits_{i \in G_j} \bm w_i - n_j \sum\limits_{i \notin G_j} \bm w_i + K n_j \sum\limits_{g=2, g \neq j}^{K} n_g \bm \delta_g \right]$.
Hence, formula \eqref{derivation1} is equal to
\begin{align*}
\frac{K^2 n_j (n - n_j)}{n} (\tilde{y} - \bm C^{1/2} \bm \delta_j)'(\tilde{y} - \bm C^{1/2} \bm \delta_j) + \lambda_2 | \bm \delta_j |_1 + C.
\end{align*}

\begin{lemma}{}{\label{lemma:max}}
Suppose a random vector $(\bm x', \bm y')' \sim N(\bm \mu, \bm \Sigma)$, where $\bm x$ and $\bm y$ are multivariate variables.
For a given value of $\bm x$, then $\bm y = \bm \mu_{Y} + \bm \Sigma_{XY}' \bm \Sigma_{X}^{-1} (\bm x - \bm \mu_{X})$ maximizes exp$\{ -\frac{1}{2}[(\bm x', \bm y') - \bm \mu'] \bm \Sigma^{-1} [(\bm x', \bm y')' - \bm \mu] \}$, where
$\bm \mu = \left[
\begin{array}{cc}
\bm \mu_{X} \\
\bm \mu_{Y}
\end{array}\right] \mbox{and} ~
\bm \Sigma= \left[
\begin{array}{cc}
\bm \Sigma_{X},      & \bm \Sigma_{XY} \\
\bm \Sigma_{XY}',   & \bm \Sigma_{Y}
\end{array}\right]$.
\end{lemma}

\begin{proof}
We need to search for $\bm y$ to minimize $[(\bm x', \bm y') - \bm \mu'] \bm \Omega [(\bm x', \bm y')' - \bm \mu]$, where $\bm \Omega = \bm \Sigma^{-1} = \left[
\begin{array}{cc}
\bm \Omega_{X},      & \bm \Omega_{XY} \\
\bm \Omega_{XY}',   & \bm \Omega_{Y}
\end{array}\right]$.
That is, we minimize
\begin{align*}
L(\bm y) &= (\bm x', \bm y') \bm \Omega (\bm x', \bm y')' - 2\bm \mu' \bm \Omega (\bm x', \bm y')'  \\
&= (\bm x', \bm y') \left[
\begin{array}{cc}
\bm \Omega_{X},      & \bm \Omega_{XY} \\
\bm \Omega_{XY}',   & \bm \Omega_{Y}
\end{array}\right]
\left[
\begin{array}{cc}
\bm x \\
\bm y
\end{array}\right] - 2 (\bm \mu_{X}', \bm \mu_{Y}') \left[
\begin{array}{cc}
\bm \Omega_{X},      & \bm \Omega_{XY} \\
\bm \Omega_{XY}',   & \bm \Omega_{Y}
\end{array}\right]
\left[
\begin{array}{cc}
\bm x \\
\bm y
\end{array}\right]        \\
&= 2 \bm x' \bm \Omega_{XY} \bm y + \bm y' \bm \Omega_{Y} \bm y - 2 (\bm \mu_{X}' \bm \Omega_{XY} + \bm \mu_{Y}' \bm \Omega_{Y}) \bm y + C,
\end{align*}
where $C$ is a constant not depending on $\bm y$.
Taking derivative of $L(\bm y)$ and setting to zero yields
\begin{align*}
\frac{\partial L(\bm y)}{\partial \bm y} &= 2 \bm \Omega_{XY}' \bm x + 2 \bm \Omega_{Y} \bm y - 2 (\bm \Omega_{XY}' \bm \mu_{X} + \bm \Omega_{Y} \bm \mu_{Y}) = 0 \\
\bm y &= \bm \mu_{Y} - \bm \Omega_{Y}^{-1} \bm \Omega_{XY}' (\bm x - \bm \mu_{X}).
\end{align*}
This, together with a property of block matrix that $\bm \Omega_{XY}' = - \bm \Omega_{Y} \bm \Sigma_{XY}' \bm \Sigma_{X}^{-1}$, completes the proof.
\end{proof}

For a new observation $\bm x$, let $y_1 = \mu_{1y} + \bm \Sigma_{Xy}' \bm \Sigma_{X}^{-1} (\bm x - \bm \mu_{1X})$ and $y_2 = \mu_{2y} + \bm \Sigma_{Xy}' \bm \Sigma_{X}^{-1} (\bm x - \bm \mu_{2X})$.
Denote by $p_1$ = $p(\bm W = (\bm x', y_1)' | G_1)$
and $p_2$ = $p(\bm W = (\bm x', y_2)' | G_2)$.
Now we prove Proposition \ref{proposition1}.
\begin{proof}{\textbf{Proof of Proposition \ref{proposition1}}.}{}

Without loss of generality, we suppose $\pi_1 p_1 >\pi_2 p_2$, then we show below that the LDA classification rule would assign $(\bm x', y_1)'$ to $G_1$.
In order to achieve this,
we only need to prove that
\begin{align}\label{eq:3}
p_2 \geq p_3 =  p(\bm W = (\bm x', y_3)' | G_2)
\end{align}
for any value of $y_3$.
That is, we need to prove $\bm W = (\bm x', y_2)'$ will maximize
the density function of $N(\bm \mu_2, \bm \Sigma)$, which is the conclusion of Lemma \ref{lemma:max}.
As a result, $\pi_1 p_1$ = $\pi_1 p(\bm W = (\bm x', y_1)' | G_1) > \pi_2 p_2 \geq$ $\pi_2 p(\bm W = (\bm x', y_1)' | G_2)$ by taking $y_3 = y_1$ in \eqref{eq:3}.
Hence,
\begin{align*}
p(\bm x \in G_1 | \bm W = (\bm x', y_1)') = \frac{\pi_1 p_1}{p(\bm W = (\bm x', y_1)')} &> \frac{\pi_2 p(\bm W = (\bm x', y_1)' | \bm x \in G_2)}{p(\bm W = (\bm x', y_1)')} \\
&= p(\bm x \in G_2 | \bm W = (\bm x', y_1)'),
\end{align*}
implying that the LDA assigns $(\bm x', y_1)'$ to $G_1$.
\end{proof}

\begin{proposition}{}{\label{proposition2}}
For an observation $\bm x$, let $y_1 = \mu_{1y} + \bm \Sigma_{Xy}' \bm \Sigma_{X}^{-1} (\bm x - \bm \mu_{1X})$ and $y_2 = \mu_{2y} + \bm \Sigma_{Xy}' \bm \Sigma_{X}^{-1} (\bm x - \bm \mu_{2X})$.
Denote by $p_1$ = $p(\bm W = (\bm x', y_1)' | G_1)$ and $p_2$ = $p(\bm W = (\bm x', y_2)' | G_2)$.
Then $p(\bm x \in G_1 | \bm X = \bm x) > p(\bm x \in G_2 | \bm X = \bm x)$ is equivalent to $\pi_1 p_1 >\pi_2 p_2$.
\end{proposition}

\begin{proof} {}{}
Since $p(\bm x \in G_1 | \bm X = \bm x) > p(\bm x \in G_2 | \bm X = \bm x)$, we have
\begin{align}\label{propproof1}
\pi_1 p(\bm X = \bm x | \bm x \in G_1) &> \pi_2 p(\bm X = \bm x | \bm x \in G_2)   \nonumber \\
\pi_1 \exp \{ -\frac{1}{2} (\bm x - \bm \mu_{1X})' \bm \Sigma_{X}^{-1} (\bm x - \bm \mu_{1X}) \} &> \pi_2 \exp \{ -\frac{1}{2} (\bm x - \bm \mu_{2X})' \bm \Sigma_{X}^{-1} (\bm x - \bm \mu_{2X}) \} \nonumber \\
\ln \pi_1 - \frac{1}{2}(\bm x - \bm \mu_{1X})' \bm \Sigma_{X}^{-1} (\bm x - \bm \mu_{1X}) &> \ln \pi_2 - \frac{1}{2}(\bm x - \bm \mu_{2X})' \bm \Sigma_{X}^{-1} (\bm x - \bm \mu_{2X}).
\end{align}
On the other hand, $\pi_1 p_1 > \pi_2 p_2$ yields
\begin{align}\label{propproof2}
\ln \pi_1 - \frac{1}{2} \left[ \left(\begin{array}{cc}
\bm x       \\
    y_{1}
\end{array} \right) -
\left(\begin{array}{cc}
\bm \mu_{1X}       \\
    \mu_{1y}
\end{array} \right)
\right]'
\left[\begin{array}{cc}
\bm \Sigma_{X},      & \bm \Sigma_{Xy} \\
\bm \Sigma_{Xy}',   & \sigma_{y}^2
\end{array}
\right]^{-1}
\left[ \left(\begin{array}{cc}
\bm x       \\
    y_{1}
\end{array} \right) -
\left(\begin{array}{cc}
\bm \mu_{1X}       \\
    \mu_{1y}
\end{array} \right)
\right] > \nonumber \\
\ln \pi_2 - \frac{1}{2}\left[ \left(\begin{array}{cc}
\bm x       \\
    y_{2}
\end{array} \right) -
\left(\begin{array}{cc}
\bm \mu_{2X}       \\
    \mu_{2y}
\end{array} \right)
\right]'
\left[\begin{array}{cc}
\bm \Sigma_{X},      & \bm \Sigma_{Xy} \\
\bm \Sigma_{Xy}',   & \sigma_{y}^2
\end{array}
\right]^{-1}
\left[ \left(\begin{array}{cc}
\bm x       \\
    y_{2}
\end{array} \right) -
\left(\begin{array}{cc}
\bm \mu_{2X}       \\
    \mu_{2y}
\end{array} \right)
\right].
\end{align}
Now we prove Equations \eqref{propproof1} and \eqref{propproof2} are equivalent.
Since
\begin{align*}
\left[\begin{array}{cc}
\bm \Sigma_{X},      & \bm \Sigma_{Xy} \\
\bm \Sigma_{Xy}',   & \sigma_{y}^2
\end{array}\right]^{-1}
&=
\left[\begin{array}{cc}
\bm \Sigma_{X}^{-1}+\frac{\bm \Sigma_{X}^{-1} \bm \Sigma_{Xy} \bm \Sigma_{Xy}' \bm \Sigma_{X}^{-1}}{\sigma_{y}^2 - \bm \Sigma_{Xy}' \bm \Sigma_{X}^{-1} \bm \Sigma_{Xy}},      & - \frac{\bm \Sigma_{X}^{-1} \bm \Sigma_{Xy}}{\sigma_{y}^2 - \bm \Sigma_{Xy}' \bm \Sigma_{X}^{-1} \bm \Sigma_{Xy}} \\
-\frac{\bm \Sigma_{Xy}' \bm \Sigma_{X}^{-1}}{\sigma_{y}^2 - \bm \Sigma_{Xy}' \bm \Sigma_{X}^{-1} \bm \Sigma_{Xy}},   & \frac{1}{\sigma_{y}^2 - \bm \Sigma_{Xy}' \bm \Sigma_{X}^{-1} \bm \Sigma_{Xy}}
\end{array}\right] \\
&= \left[\begin{array}{cc}
\bm \Sigma_{X}^{-1}+\frac{\bm \Sigma_{X}^{-1} \bm \Sigma_{Xy} \bm \Sigma_{Xy}' \bm \Sigma_{X}^{-1}}{Var(y|\bm X)},      & - \frac{\bm \Sigma_{X}^{-1} \bm \Sigma_{Xy}}{Var(y|\bm X)} \\
-\frac{\bm \Sigma_{Xy}' \bm \Sigma_{X}^{-1}}{Var(y|\bm X)},   & \frac{1}{Var(y|\bm X)}
\end{array}\right],
\end{align*}
the left side of \eqref{propproof2} equals
\begin{align*}
&~~~\ln \pi_1 - \frac{1}{2} \left[ (\bm x - \bm \mu_{1X})', y_1 - \mu_{1y} \right]'
\left[\begin{array}{cc}
\bm \Sigma_{X}^{-1}+\frac{\bm \Sigma_{X}^{-1} \bm \Sigma_{Xy} \bm \Sigma_{Xy}' \bm \Sigma_{X}^{-1}}{Var(y|\bm X)},      & - \frac{\bm \Sigma_{X}^{-1} \bm \Sigma_{Xy}}{Var(y|\bm X)} \\
-\frac{\bm \Sigma_{Xy}' \bm \Sigma_{X}^{-1}}{Var(y|\bm X)},   & \frac{1}{Var(y|\bm X)}
\end{array}\right]
\left[ \begin{array}{cc}
\bm x - \bm \mu_{1X}      \\
y_1 - \mu_{1y}
\end{array}
\right] \\
&= \ln \pi_1 - \frac{1}{2} \{(\bm x - \bm \mu_{1X})' \bm \Sigma_{X}^{-1} (\bm x - \bm \mu_{1X}) + (\bm x - \bm \mu_{1X})' \frac{\bm \Sigma_{X}^{-1} \bm \Sigma_{Xy} \bm \Sigma_{Xy}' \bm \Sigma_{X}^{-1}}{Var(y|\bm X)} (\bm x - \bm \mu_{1X}) \\
&~~~- \frac{y_1 - \mu_{1y}}{Var(y|\bm X)} \bm \Sigma_{Xy}' \bm \Sigma_{X}^{-1} (\bm x - \bm \mu_{1X}) + \frac{(y_1 - \mu_{1y})^2}{Var(y|\bm X)} - (\bm x - \bm \mu_{1X})' \frac{\bm \Sigma_{X}^{-1} \bm \Sigma_{Xy}}{Var(y|\bm X)} (y_1 - \mu_{1y}) \} \\
&= \ln \pi_1 - \frac{1}{2} (\bm x - \bm \mu_{1X})' \bm \Sigma_{X}^{-1} (\bm x - \bm \mu_{1X}),
\end{align*}
where the last equality applies $y_1 - \mu_{1y} = \bm \Sigma_{Xy}' \bm \Sigma_{X}^{-1} (\bm x - \bm \mu_{1X})$. Similarly, the right side of Equation \eqref{propproof2} equals $\ln \pi_2 - \frac{1}{2}(\bm x - \bm \mu_{2X})' \bm \Sigma_{X}^{-1} (\bm x - \bm \mu_{2X})$. This completes the proof.
\end{proof}

The inequality $p(\bm x \in G_1 | \bm X = \bm x) > p(\bm x \in G_2 | \bm X = \bm x)$ in Proposition \ref{proposition2} indicates that the LDA rule assigns $\bm x$ to $G_1$. Therefore, Proposition \ref{proposition2} implies that {\bf Step 2b} of Algorithm \ref{alg2} is equivalent to applying the LDA classification rule directly on $\bm x$ instead of $\bm w = (\bm x', \hat{y})'$.
This fact enables us to give theoretical proof for the consistency properties of the proposed classification rule based on variable $\bm X$ rather than $\bm W = (\bm X', y)'$.
Before the proof of Theorem \ref{theory:multiclass}, we present Lemmas \ref{lemma3} - \ref{lemma4}.

\begin{lemma}{}{\label{lemma3}}
For any $k = 2, 3, \ldots, K$, we have
\begin{equation*}
(\hat{\bm \delta}'_{kX} \hat{\bm C}_X - (\bm \delta_{kX}^0)' \bm C_X^0) \bm \Sigma_X^0 (\hat{\bm C}_X \hat{\bm \delta}_{kX} - \bm C_X^0 \bm \delta_{kX}^0) = \Delta_k^2 \left[ O_p(\frac{b_k^{(n)}}{\Delta_k}) + O_p(d_n) \right]
\end{equation*}
for the multi-class problem.
\end{lemma}

\begin{proof}{}{}
Decompose
\begin{equation}\label{eq:4}
(\hat{\bm C}_X \hat{\bm \delta}_{kX} - \bm C_X^0 \bm \delta_{kX}^0)'
\bm \Sigma_X^0 (\hat{\bm C}_X \hat{\bm \delta}_{kX} - \bm C_X^0 \bm \delta_{kX}^0) = \hat{\bm \delta}'_{kX} \hat{\bm C}_X \bm \Sigma_X^0 \hat{\bm C}_X \hat{\bm \delta}_{kX} - 2 \hat{\bm \delta}'_{kX} \hat{\bm C}_X \bm \delta_{kX}^0 + (\bm \delta_{kX}^0)' \bm C_X^0 \bm \delta_{kX}^0.
\end{equation}
On one hand, by the result \eqref{conrateC} we have
\begin{equation*}
\hat{\bm \delta}'_{kX} \hat{\bm C}_X \bm \Sigma_X^0 \hat{\bm C}_X \hat{\bm \delta}_{kX} = \hat{\bm \delta}'_{kX} \hat{\bm C}_X \hat{\bm \delta}_{kX} [1 + O_p(d_n)] = \hat{\bm \delta}'_{kX} \bm C_X^0 \hat{\bm \delta}_{kX} [1 + O_p(d_n)].
\end{equation*}
Since $E[ (\bm \delta^0_{kX})' \bm C_X^0 (\hat{\bm \delta}_{kX} - \bm \delta^0_{kX})]^2 \leq \Delta_k^2 E[(\hat{\bm \delta}_{kX} - \bm \delta^0_{kX})' \bm C_X^0 (\hat{\bm \delta}_{kX} - \bm \delta_{kX}^0)]$ and by Equation \eqref{conratedelta}, we obtain
\begin{align*}
\hat{\bm \delta}'_{kX} \bm C_X^0 \hat{\bm \delta}_{kX} &= (\bm \delta_{kX}^0)' \bm C_X^0 \bm \delta_{kX}^0 + 2 (\bm \delta_{kX}^0)' \bm C_X^0 (\hat{\bm \delta}_{kX} - \bm \delta_{kX}^0) + (\hat{\bm \delta}_{kX} - \bm \delta_{kX}^0)' \bm C_X^0 (\hat{\bm \delta}_{kX} - \bm \delta_{kX}^0) \\
&= \Delta_k^2 + O_p(b_k^{(n)} \Delta_k) + O_p((b_k^{(n)})^2) \\
&= \Delta_k^2 [1 + O_p(\frac{b_k^{(n)}}{\Delta_k}) + O_p(\frac{(b_k^{(n)})^2}{\Delta_k^2})] \\
&= \Delta_k^2 [1 + O_p(\frac{b_k^{(n)}}{\Delta_k})].
\end{align*}
As a result,
\begin{equation}\label{eq:5}
\hat{\bm \delta}'_{kX} \hat{\bm C}_X \hat{\bm \delta}_{kX}
= \hat{\bm \delta}'_{kX} \bm C_X^0 \hat{\bm \delta}_{kX} [1 + O_p(d_n)]
= \Delta_k^2 [1 + O_p(\frac{b_k^{(n)}}{\Delta_k}) + O_p(d_n)].
\end{equation}
On the other hand, since $\| \bm \delta_{kX}^0 \|_2^2 = O(\Delta_k^2)$, we have
\begin{align}\label{eq:insert1}
(\bm \delta_{kX}^0)' \hat{\bm C}_X \bm \delta_{kX}^0 = (\bm \delta_{kX}^0)' (\hat{\bm C}_X - \bm C_X^0) \bm \delta_{kX}^0 + (\bm \delta_{kX}^0)' \bm C_X^0 \bm \delta_{kX}^0
= O_p(\Delta_k^2 d_n) + \Delta_k^2
= \Delta_k^2 [1 + O_p(d_n)].
\end{align}
Consequently,
\begin{align}\label{eq:6}
\hat{\bm \delta}'_{kX} \hat{\bm C}_X \bm \delta_{kX}^0 &= \Delta_k \sqrt{1 + O_p(d_n)}~\Delta_k \sqrt{1 + O_p(\frac{b_k^{(n)}}{\Delta_k}) + O_p(d_n)}      \nonumber   \\
&= \Delta_k^2 \sqrt{1 + O_p(\frac{b_k^{(n)}}{\Delta_k}) + O_p(d_n)}.
\end{align}
Combing Equations \eqref{eq:4}, \eqref{eq:5} and \eqref{eq:6} yields
\begin{align*}
&(\hat{\bm \delta}'_{kX} \hat{\bm C}_X - (\bm \delta_{kX}^0)' \bm C_X^0) \bm \Sigma_X^0 (\hat{\bm C}_X \hat{\bm \delta}_{kX} - \bm C_X^0 \bm \delta_{kX}^0)   \\
=& \Delta_k^2 [1 + O_p(\frac{b_k^{(n)}}{\Delta_k}) + O_p(d_n)] - 2 \Delta_k^2 \sqrt{1 + O_p(\frac{b_k^{(n)}}{\Delta_k}) + O_p(d_n)} + \Delta_k^2 \\
=& \Delta_k^2 \left[ O_p(\frac{b_k^{(n)}}{\Delta_k}) + O_p(d_n) \right],
\end{align*}
where the last equality uses the Taylor expansion of $\sqrt{1 + x} = 1 + \frac{1}{2} x + o(x)$.
\end{proof}

Write $\bm \mu_k^0 = ((\bm \mu_{kX}^0)', \mu_{ky}^0)'$, where $\bm \mu_{kX}^0$ is the true mean value of variable $\bm X$ for class $G_k$.
Correspondingly, write $\hat{\bm \mu}_k = (\hat{\bm \mu}_{kX}', \hat{ \mu}_{ky})'$. Let $a_n \asymp b_n$ represent two sequences $a_n$ and $b_n$ to be the same order.
Now we state Lemma \ref{lemma5}.

\begin{lemma}{}{\label{lemma5}}
Let $q_k^{(n)}$ be the number of nonzero entries of estimate $\hat{\bm \delta}_{kX}$.
For $k = 2, 3, \ldots, K$, we have
\begin{align*}
\hat{\bm \delta}'_{kX} \hat{\bm C}_X (\hat{\bm \mu}_{1X} - \bm \mu_{1X}^0) & \asymp \hat{\bm \delta}'_{kX} \hat{\bm C}_X (\hat{\bm \mu}_{kX} - \bm \mu_{kX}^0) \\
&= O_p(\sqrt{\frac{q_k^{(n)}}{n}}) \sqrt{\hat{\bm \delta}'_{kX} \hat{\bm C}_X \hat{\bm \delta}_{kX}} - O_p(\sqrt{\frac{S_{h;p} q_k^{(n)}}{n}}) \sqrt{\hat{\bm \delta}'_{kX} \hat{\bm C}_X \hat{\bm \delta}_{kX}}.
\end{align*}
\end{lemma}

\begin{proof}{}{}
Without loss of generality, we assume that $\hat{\bm \delta}_{kX} = (\hat{\bm \delta}'_{k,1}, \bm 0')'$, where $\hat{\bm \delta}'_{k,1}$ is a $q_k^{(n)}$-dimensional vector containing all the nonzero entries of $\hat{\bm \delta}_{kX}$.
Note that $\lim\limits_{n \rightarrow \infty} q_k^{(n)} = s_k$.
Conformally, we write
\begin{align*}
\bm \Sigma_X^0 = \left[
\begin{array}{cc}
\bm \Sigma^0_{11},      & \bm \Sigma^0_{12} \\
(\bm \Sigma^0_{12})',   & \bm \Sigma^0_{22}
\end{array}\right],~~~
\hat{\bm \Sigma}_X = \left[
\begin{array}{cc}
\hat{\bm \Sigma}_{11},      & \hat{\bm \Sigma}_{12} \\
(\hat{\bm \Sigma}_{12})',   & \hat{\bm \Sigma}_{22}
\end{array}\right], \\
\bm C_X^0 = \left[
\begin{array}{cc}
\bm C^0_{11},      & \bm C^0_{12} \\
(\bm C^0_{12})',   & \bm C^0_{22}
\end{array}\right],~~
\hat{\bm C}_X = \left[
\begin{array}{cc}
\hat{\bm C}_{11},      & \hat{\bm C}_{12} \\
(\hat{\bm C}_{12})',   & \hat{\bm C}_{22}
\end{array}\right],
\end{align*}
where $\bm \Sigma^0_{11}, \hat{\bm \Sigma}_{11}, \bm C^0_{11}$ and $\hat{\bm C}_{11}$ are $q_k^{(n)} \times q_k^{(n)}$ matrices.
Let $\hat{\bm \mu}_{1X} - \bm \mu_{1X}^0 = (\bm \eta'_1, \bm \eta'_2)'$ with $\bm \eta_1$ a $q_k^{(n)}$-dimensional vector.
Hence,
\begin{equation*}
\hat{\bm \delta}'_{kX} \hat{\bm C}_X (\hat{\bm \mu}_{1X} - \bm \mu_{1X}^0)
= \hat{\bm \delta}'_{k,1} \hat{\bm C}_{11} \bm \eta_1 + \hat{\bm \delta}'_{k,1} \hat{\bm C}_{12} \bm \eta_2
= \hat{\bm \delta}'_{k,1} \hat{\bm C}_{11} \bm \eta_1 - \hat{\bm \delta}'_{k,1} \hat{\bm \Sigma}_{11}^{-1} \hat{\bm \Sigma}_{12} \hat{\bm C}_{22} \bm \eta_2.
\end{equation*}
On one hand,
\begin{align*}
(\hat{\bm \delta}'_{k,1} \hat{\bm C}_{11} \bm \eta_1)^2 \leq (\hat{\bm \delta}'_{k,1} \hat{\bm C}_{11} \hat{\bm \delta}_{k,1})(\bm \eta'_1 \hat{\bm C}_{11} \bm \eta_1) &= (\hat{\bm \delta}'_{kX} \hat{\bm C}_X \hat{\bm \delta}_{kX})(\bm \eta'_1 \hat{\bm C}_{11} \bm \eta_1) \\
&= O_p(\frac{q_k^{(n)}}{n})(\hat{\bm \delta}'_{kX} \hat{\bm C}_X \hat{\bm \delta}_{kX}).
\end{align*}
On the other hand,
\begin{align*}
(\hat{\bm \delta}'_{k,1} \hat{\bm \Sigma}_{11}^{-1} \hat{\bm \Sigma}_{12} \hat{\bm C}_{22} \bm \eta_2)^2 &\leq (\hat{\bm \delta}'_{k,1} \hat{\bm \Sigma}_{11}^{-1} \hat{\bm \delta}_{k,1})(\bm \eta'_2 \hat{\bm C}_{22} \hat{\bm \Sigma}'_{12} \hat{\bm \Sigma}_{11}^{-1} \hat{\bm \Sigma}_{12} \hat{\bm C}_{22} \bm \eta_2) \\
&\leq (\hat{\bm \delta}'_{k,1} \hat{\bm C}_{11} \hat{\bm \delta}_{k,1}) (\bm \eta'_2 \hat{\bm C}_{22} \hat{\bm \Sigma}'_{12} \hat{\bm \Sigma}_{11}^{-1} \hat{\bm \Sigma}_{12} \hat{\bm C}_{22} \bm \eta_2) \\
&= (\hat{\bm \delta}'_{kX} \hat{\bm C}_X \hat{\bm \delta}_{kX})(\bm \eta'_2 \hat{\bm C}_{22} \hat{\bm \Sigma}'_{12} \hat{\bm \Sigma}_{11}^{-1} \hat{\bm \Sigma}_{12} \hat{\bm C}_{22} \bm \eta_2) \\
&= (\hat{\bm \delta}'_{kX} \hat{\bm C}_X \hat{\bm \delta}_{kX})(\bm \eta'_2 \bm C^0_{22} (\bm \Sigma^0_{12})^{'} (\bm \Sigma^0_{11})^{-1} \bm \Sigma^0_{12} \bm C^0_{22} \bm \eta_2 [1 + O_p(d_n)]) \\
& \triangleq (\hat{\bm \delta}'_{kX} \hat{\bm C}_X \hat{\bm \delta}_{kX})(\omega_n [1 + O_p(d_n)]),
\end{align*}
where the forth equation is obtained from \eqref{conrateC}, and $\omega_n = (\bm \eta'_2 \bm C^0_{22} (\bm \Sigma^0_{12})^{'} (\bm \Sigma^0_{11})^{-1} \bm \Sigma^0_{12} \bm C^0_{22} \bm \eta_2$.
Hence, we have
\begin{equation*}
\hat{\bm \delta}'_{kX} \hat{\bm C}_X (\hat{\bm \mu}_{1X} - \bm \mu_{1X}^0)
= O_p(\sqrt{\frac{q_k^{(n)}}{n}}) \sqrt{\hat{\bm \delta}'_{kX} \hat{\bm C}_X \hat{\bm \delta}_{kX}} - \sqrt{\hat{\bm \delta}'_{kX} \hat{\bm C}_X \hat{\bm \delta}_{kX}} \sqrt{\omega_n [1 + O_p(d_n)]}.
\end{equation*}
Under condition (C1),
\begin{align*}
E(\omega_n) \leq \theta E(\bm \eta'_2 \bm C^0_{22} (\bm \Sigma^0_{12})^{'} \bm \Sigma^0_{12} \bm C^0_{22} \bm \eta_2)
&= \frac{\theta}{n} \mbox{tr} [\bm \Sigma^0_{12} \bm C^0_{22} \bm \Sigma^0_{22} \bm C^0_{22} (\bm \Sigma^0_{12})^{'}]  \\
&\leq \frac{\theta^4}{n} \mbox{tr} [\bm \Sigma^0_{12} (\bm \Sigma^0_{12})^{'}].
\end{align*}
Recall that $\bm \Sigma^0 = (\sigma_{ij}^0)_{1 \leq i, j \leq p}$, then
\begin{align*}
E(\omega_n) \leq \frac{\theta^4}{n} \sum_{i=1}^{q_k^{(n)}} \sum_{j=q_k^{(n)} + 1}^{p} (\sigma_{ij}^0)^2
&\leq \frac{\theta^4}{n} q_k^{(n)} \max_{i} \sum_{j=q_k^{(n)} + 1}^{p} (\sigma_{ij}^0)^2 \\
&\leq \frac{\theta^{6-h}}{n} q_k^{(n)} \max_{j \leq p} \sum_{i=1}^{p} |\sigma_{ij}^0|^h \\
&= O(\frac{S_{h;p}q_k^{(n)}}{n}).
\end{align*}
Consequently,
\begin{equation*}
\hat{\bm \delta}'_{kX} \hat{\bm C}_X (\hat{\bm \mu}_{1X} - \bm \mu_{1X}^0)
= O_p(\sqrt{\frac{q_k^{(n)}}{n}}) \sqrt{\hat{\bm \delta}'_{kX} \hat{\bm C}_X \hat{\bm \delta}_{kX}} - O_p(\sqrt{\frac{S_{h;p} q_k^{(n)}}{n}}) \sqrt{\hat{\bm \delta}'_{kX} \hat{\bm C}_X \hat{\bm \delta}_{kX}}.
\end{equation*}
Similarly, we have
\begin{equation*}
\hat{\bm \delta}'_{kX} \hat{\bm C}_X (\hat{\bm \mu}_{kX} - \bm \mu_{kX}^0)
= O_p(\sqrt{\frac{q_k^{(n)}}{n}}) \sqrt{\hat{\bm \delta}'_{kX} \hat{\bm C}_X \hat{\bm \delta}_{kX}} - O_p(\sqrt{\frac{S_{h;p} q_k^{(n)}}{n}}) \sqrt{\hat{\bm \delta}'_{kX} \hat{\bm C}_X \hat{\bm \delta}_{kX}}.
\end{equation*}
\end{proof}

\begin{lemma}{}{\label{lemma4}}
For $t = 1, 2, \ldots, K$ and $k = 2, 3, \ldots, K$, we have
\begin{align*}
&(\bm \mu_{tX}^0)' (\hat{\bm C}_X \hat{\bm \delta}_{kX} - \bm C_X^0 \bm \delta_{kX}^0) - (\frac{\hat{\bm \mu}_{1X} + \hat{\bm \mu}_{kX}}{2})' \hat{\bm C}_X \hat{\bm \delta}_{kX} + (\frac{\bm \mu_{1X}^0 + \bm \mu_{kX}^0}{2})' \bm C_X^0 \bm \delta_{kX}^0 \\
=& \Delta^2 \left[ O_p(\frac{b_k^{(n)}}{\Delta_k}) + O_p(d_n) + O_p(\frac{\sqrt{S_{h;p} q_k^{(n)}}}{\sqrt{n} \Delta_k}) \right]  + \frac{1}{2} (\bm \delta_{kX}^0 - \bm \delta_{tX}^0)' \bm C_X^0 (\bm \delta_{kX}^0 - \bm \delta_{tX}^0)
\end{align*}
for the multi-class problem.
\end{lemma}

\begin{proof}{}{}
It is not difficult to derive
\begin{align*}
& (\bm \mu_{tX}^0)' (\hat{\bm C}_X \hat{\bm \delta}_{kX} - \bm C_X^0 \bm \delta_{kX}^0) - (\frac{\hat{\bm \mu}_{1X} + \hat{\bm \mu}_{kX}}{2})' \hat{\bm C}_X \hat{\bm \delta}_{kX} + (\frac{\bm \mu_{1X}^0 + \bm \mu_{kX}^0}{2})' \bm C_X^0 \bm \delta_{kX}^0  \\
=& (\bm \mu_{tX}^0)' (\hat{\bm C}_X \hat{\bm \delta}_{kX} - \bm C_X^0 \bm \delta_{kX}^0) - \frac{1}{2} [(\hat{\bm \mu}_{1X} + \hat{\bm \mu}_{kX}) - (\bm \mu_{1X}^0 + \bm \mu_{kX}^0)]' \hat{\bm C}_X \hat{\bm \delta}_{kX} \\
&~~~~~~~~~~~~~~~~~~~~~~~~~~~~~~~~~~~ - \frac{1}{2} (\bm \mu_{1X}^0 + \bm \mu_{kX}^0)' (\hat{\bm C}_X \hat{\bm \delta}_{kX} - \bm C_X^0 \bm \delta_{kX}^0)    \\
=& [\frac{(\bm \mu_{tX}^0 - \bm \mu_{1X}^0) + (\bm \mu_{tX}^0 - \bm \mu_{kX}^0)}{2}]' (\hat{\bm C}_X \hat{\bm \delta}_{kX} - \bm C_X^0 \bm \delta_{kX}^0) - \frac{1}{2} [(\hat{\bm \mu}_{1X} - \bm \mu_{1X}^0) + (\hat{\bm \mu}_{kX} - \bm \mu_{kX}^0)]' \hat{\bm C}_X \hat{\bm \delta}_{kX}   \\
=& K (\bm \delta_{tX}^0)' (\hat{\bm C}_X \hat{\bm \delta}_{kX} - \bm C_X^0 \bm \delta_{kX}^0) - \frac{K}{2} (\bm \delta_{kX}^0)' (\hat{\bm C}_X \hat{\bm \delta}_{kX} - \bm C_X^0 \bm \delta_{kX}^0) \\
&~~~~~~~~~~~~~~~~~~~~~~~~~~~~~~~~~~~ - \frac{1}{2} [(\hat{\bm \mu}_{1X} - \bm \mu_{1X}^0) + (\hat{\bm \mu}_{kX} - \bm \mu_{kX}^0)]' \hat{\bm C}_X \hat{\bm \delta}_{kX}.
\end{align*}
Because
\begin{align*}
(\bm \delta_{kX}^0 - \bm \delta_{tX}^0)' \bm C_X^0 (\bm \delta_{kX}^0 - \bm \delta_{tX}^0) = \Delta_k^2 + \Delta_t^2 - 2 (\bm \delta_{tX}^0)' \bm C_X^0 \bm \delta_{kX}^0,
\end{align*}
we hence have
\begin{align*}
- (\bm \delta_{tX}^0)' \bm C_X^0 \bm \delta_{kX}^0 &= \frac{1}{2} (\bm \delta_{kX}^0 - \bm \delta_{tX}^0)' \bm C_X^0 (\bm \delta_{kX}^0 - \bm \delta_{tX}^0) - \frac{1}{2} (\Delta_k^2 + \Delta_t^2)  \\
& \leq \frac{1}{2} (\bm \delta_{kX}^0 - \bm \delta_{tX}^0)' \bm C_X^0 (\bm \delta_{kX}^0 - \bm \delta_{tX}^0) - \Delta_k \Delta_t.
\end{align*}
Consequently, applying the Cauchy-Schwarz inequality together with Equations \eqref{eq:5} and \eqref{eq:insert1}, we can obtain
\begin{align*}
& (\bm \delta_{tX}^0)' (\hat{\bm C}_X \hat{\bm \delta}_{kX} - \bm C_X^0 \bm \delta_{kX}^0) \\
\leq &  \Delta_t \sqrt{1 + O_p(d_n)}~\Delta_k \sqrt{1 + O_p(\frac{b_k^{(n)}}{\Delta_k}) + O_p(d_n)} + \frac{1}{2} (\bm \delta_{kX}^0 - \bm \delta_{tX}^0)' \bm C_X^0 (\bm \delta_{kX}^0 - \bm \delta_{tX}^0) - \Delta_t \Delta_k  \\
\leq & \Delta_t \Delta_k (1 + O_p(\frac{b_k^{(n)}}{\Delta_k}) + O_p(d_n)) - \Delta_t \Delta_k + \frac{1}{2} (\bm \delta_{kX}^0 - \bm \delta_{tX}^0)' \bm C_X^0 (\bm \delta_{kX}^0 - \bm \delta_{tX}^0)  \\
= & \Delta_t \Delta_k (O_p(\frac{b_k^{(n)}}{\Delta_k}) + O_p(d_n)) + \frac{1}{2} (\bm \delta_{kX}^0 - \bm \delta_{tX}^0)' \bm C_X^0 (\bm \delta_{kX}^0 - \bm \delta_{tX}^0)   \\
\leq & \Delta^2 (O_p(\frac{b_k^{(n)}}{\Delta_k}) + O_p(d_n)) + \frac{1}{2} (\bm \delta_{kX}^0 - \bm \delta_{tX}^0)' \bm C_X^0 (\bm \delta_{kX}^0 - \bm \delta_{tX}^0).
\end{align*}
Similarly,
\begin{align*}
(\bm \delta_{kX}^0)' (\hat{\bm C}_X \hat{\bm \delta}_{kX} - \bm C_X^0 \bm \delta_{kX}^0)
= \Delta_k^2 (O_p(\frac{b_k^{(n)}}{\Delta_k}) + O_p(d_n))
\leq \Delta^2 (O_p(\frac{b_k^{(n)}}{\Delta_k}) + O_p(d_n)).
\end{align*}
As a result, according to Lemma \ref{lemma5}, we have
\begin{align*}
& (\bm \mu_{tX}^0)' (\hat{\bm C}_X \hat{\bm \delta}_{kX} - \bm C_X^0 \bm \delta_{kX}^0) - (\frac{\hat{\bm \mu}_{1X} + \hat{\bm \mu}_{kX}}{2})' \hat{\bm C}_X \hat{\bm \delta}_{kX} + (\frac{\bm \mu_{1X}^0 + \bm \mu_{kX}^0}{2})' \bm C_X^0 \bm \delta_{kX}^0   \\
\leq & \Delta^2 (O_p(\frac{b_k^{(n)}}{\Delta_k}) + O_p(d_n)) + \frac{1}{2} (\bm \delta_{kX}^0 - \bm \delta_{tX}^0)' \bm C_X^0 (\bm \delta_{kX}^0 - \bm \delta_{tX}^0) \\
&~~~~~~~~~~~~~~~~~~~~~~~~~~~~~~ + \left( O_p(\sqrt{\frac{S_{h;p} q_k^{(n)}}{n}}) -  O_p(\sqrt{\frac{q_k^{(n)}}{n}}) \right) \sqrt{\hat{\bm \delta}'_{kX} \hat{\bm C}_X \hat{\bm \delta}_{kX}}  \\
\leq & \Delta^2 (O_p(\frac{b_k^{(n)}}{\Delta_k}) + O_p(d_n)) + \Delta_k O_p(\sqrt{\frac{S_{h;p} q_k^{(n)}}{n}}) \sqrt{1 + O_p(\frac{b_k^{(n)}}{\Delta_k}) + O_p(d_n)} \\
&~~~~~~~~~~~~~~~~~~~~~~~~~~~~~~ + \frac{1}{2} (\bm \delta_{kX}^0 - \bm \delta_{tX}^0)' \bm C_X^0 (\bm \delta_{kX}^0 - \bm \delta_{tX}^0)  \\
= & \Delta^2 (O_p(\frac{b_k^{(n)}}{\Delta_k}) + O_p(d_n)) + \Delta_k^2 ( O_p(\frac{b_k^{(n)}}{\Delta_k^2}) + O_p(\frac{d_n}{\Delta_k}) + O_p(\frac{\sqrt{S_{h;p} q_k^{(n)}}}{\sqrt{n} \Delta_k}) ) \\
&~~~~~~~~~~~~~~~~~~~~~~~~~~~~~~ + \frac{1}{2} (\bm \delta_{kX}^0 - \bm \delta_{tX}^0)' \bm C_X^0 (\bm \delta_{kX}^0 - \bm \delta_{tX}^0)   \\
\leq & \Delta^2 \left( O_p(\frac{b_k^{(n)}}{\Delta_k}) + O_p(d_n) + O_p(\frac{\sqrt{S_{h;p} q_k^{(n)}}}{\sqrt{n} \Delta_k}) \right) + \frac{1}{2} (\bm \delta_{kX}^0 - \bm \delta_{tX}^0)' \bm C_X^0 (\bm \delta_{kX}^0 - \bm \delta_{tX}^0).
\end{align*}
\end{proof}
From Lemma \ref{lemma4}, note that when $t = k$, we have $(\bm \mu_{kX}^0)' (\hat{\bm C}_X \hat{\bm \delta}_{kX} - \bm C_X^0 \bm \delta_{kX}^0) - (\frac{\hat{\bm \mu}_{1X} + \hat{\bm \mu}_{kX}}{2})' \hat{\bm C}_X \hat{\bm \delta}_{kX}\\ + (\frac{\bm \mu_{1X}^0 + \bm \mu_{kX}^0}{2})' \bm C_X^0 \bm \delta_{kX}^0 = \Delta^2 \left[ O_p(\frac{b_k^{(n)}}{\Delta_k}) + O_p(d_n) + O_p(\frac{\sqrt{S_{h;p} q_k^{(n)}}}{\sqrt{n} \Delta_k}) \right]$.
With Lemmas \ref{lemma3} - \ref{lemma4}, we are ready to complete the proof of Theorem \ref{theory:multiclass}.

\begin{proof}{\textbf{Proof of Theorem \ref{theory:multiclass}}.}{}

Let $\hat{Z}_{PROP}$ and $\hat{Z}_{Bayes}$ denote the predicted class labels obtained by the proposed model and the Bayes rule, respectively.
For simplicity, we assume $\pi_1 = \pi_2 = \ldots = \pi_K$ instead of condition (C7) in the proofs of Theorems \ref{theory:multiclass} - \ref{theory:supp} with no influence on the theoretical results, since condition (C7) is only used to bound the term $\log \frac{\pi_k}{\pi_1}$ of the LDA rule.
Define $\vartheta_k = (\bm x - \frac{\bm \mu_{1X}^0 + \bm \mu_{kX}^0}{2})' \bm C_X^0 \bm \delta_{kX}^0$ and $\hat{\vartheta}_k = (\bm x - \frac{\hat{\bm \mu}_{1X} + \hat{\bm \mu}_{kX}}{2})' \hat{\bm C}_X \hat{\bm \delta}_{kX}$ for a new sample $\bm x$.
Then for any $\epsilon > 0$,
\begin{align*}
R_{PROP}(\mathcal{T}) - R_{Bayes} &\leq \Pr(\hat{Z}_{PROP} \neq \hat{Z}_{Bayes})   \\
&\leq 1 - \Pr(|\hat{\vartheta}_k - \vartheta_k| < \frac{\epsilon}{2}, |\vartheta_k - \vartheta_l| > \epsilon~ \mbox{for any } k, l)        \\
&\leq \Pr(|\hat{\vartheta}_k - \vartheta_k| \geq \frac{\epsilon}{2}~ \mbox{for some } k) + \Pr(|\vartheta_k - \vartheta_l| \leq \epsilon~ \mbox{for some } k, l).
\end{align*}
Firstly, we bound the probability $\Pr(|\vartheta_k - \vartheta_l| \leq \epsilon~ \mbox{for some } k, l)$.
Since $\vartheta_k - \vartheta_l = \bm x' \bm C_X^0 (\bm \delta_{kX}^0 - \bm \delta_{lX}^0) - (\frac{\bm \mu_{1X}^0 + \bm \mu_{kX}^0}{2})' \bm C_X^0 \bm \delta_{kX}^0 + (\frac{\bm \mu_{1X}^0 + \bm \mu_l^0}{2})' \bm C_X^0 \bm \delta_{lX}^0$, the variance of $\vartheta_k - \vartheta_l$ is $(\bm \delta_{kX}^0 - \bm \delta_{lX}^0)' \bm C_X^0 (\bm \delta_{kX}^0 - \bm \delta_{lX}^0)$.
Hence,
\begin{align*}
\Pr(|\vartheta_k - \vartheta_l| \leq \epsilon~ \mbox{for some } k, l) &= \sum_{t=1}^K \Pr(|\vartheta_k - \vartheta_l| \leq \epsilon | Z = t) \pi_t  \\
&\leq \sum_{k,l,t} \pi_t \frac{C \epsilon}{\sqrt{(\bm \delta_{kX}^0 - \bm \delta_{lX}^0)' \bm C_X^0 (\bm \delta_{kX}^0 - \bm \delta_{lX}^0)}}  \\
&\leq C K^2 \epsilon,
\end{align*}
where the last inequality is obtained by condition (C6).
Secondly, we bound the term $\Pr(|\hat{\vartheta}_k - \vartheta_k| \geq \frac{\epsilon}{2}~ \mbox{for some } k)$.
As $(\hat{\vartheta}_k - \vartheta_k | Z = t) = \bm x' (\hat{\bm C}_X \hat{\bm \delta}_{kX} - \bm C_X^0 \bm \delta_{kX}^0) - (\frac{\hat{\bm \mu}_{1X} + \hat{\bm \mu}_{kX}}{2})' \hat{\bm C}_X \hat{\bm \delta}_{kX} + (\frac{\bm \mu_{1X}^0 + \bm \mu_{kX}^0}{2})' \bm C_X^0 \bm \delta_{kX}^0$, the conditional difference term
$(\hat{\vartheta}_k - \vartheta_k | Z = t)$ is from normal distribution $N( \mu_{\vartheta}, \sigma^2_{\vartheta})$ with
\begin{align*}
\mu_{\vartheta}^{(t)} = (\bm \mu_{tX}^{0})' (\hat{\bm C}_X \hat{\bm \delta}_{kX} - \bm C_X^0 \bm \delta_{kX}^0) - (\frac{\hat{\bm \mu}_{1X} + \hat{\bm \mu}_{kX}}{2})' \hat{\bm C}_X \hat{\bm \delta}_{kX} + (\frac{\bm \mu_{1X}^0 + \bm \mu_{kX}^0}{2})' \bm C_X^0 \bm \delta_{kX}^0
\end{align*}
and
\begin{align*}
\sigma^2_{\vartheta} = (\hat{\bm \delta}'_{kX} \hat{\bm C}_X - (\bm \delta_{kX}^0)' \bm C_X^0) \bm \Sigma_X^0 (\hat{\bm C}_X \hat{\bm \delta}_{kX} - \bm C_X^0 \bm \delta_{kX}^0).
\end{align*}
By Markov's inequality, together with Lemmas \ref{lemma3} and \ref{lemma4}, we have
\begin{align}\label{eq:7}
& \Pr(|\hat{\vartheta}_k - \vartheta_k| \geq \frac{\epsilon}{2}~ \mbox{for some } k) \nonumber \\
= & \sum_{t \neq k}^K \pi_t \Pr(|\hat{\vartheta}_k - \vartheta_k| \geq \frac{\epsilon}{2} | Z = t) + \pi_k \Pr(|\hat{\vartheta}_k - \vartheta_k| \geq \frac{\epsilon}{2} | Z = k) \nonumber \\
\leq & \frac{C \max\limits_k (\hat{\bm \delta}'_{kX} \hat{\bm C}_X - (\bm \delta_{kX}^0)' \bm C_X^0) \bm \Sigma_X^0 (\hat{\bm C}_X \hat{\bm \delta}_{kX} - \bm C_X^0 \bm \delta_{kX}^0)}{(\epsilon - \mu_{\vartheta}^{(t \neq k)})^2} \nonumber \\
& ~~~~~~~~~~~~~~~~~~~~~~~~~~~~~~~~ + \frac{(\hat{\bm \delta}'_{kX} \hat{\bm C}_X - (\bm \delta_{kX}^0)' \bm C_X^0) \bm \Sigma_X^0 (\hat{\bm C}_X \hat{\bm \delta}_{kX} - \bm C_X^0 \bm \delta_{kX}^0)}{(\epsilon - \mu_{\vartheta}^{(k)})^2}    \nonumber    \\
\leq & \frac{C \max\limits_k \Delta_k^2 [ O_p(\frac{b_k^{(n)}}{\Delta_k}) + O_p(d_n) ]}{\left[\epsilon - \Delta^2 ( O_p(\frac{b_k^{(n)}}{\Delta_k}) + O_p(d_n) + O_p(\frac{\sqrt{S_{h;p} q_k^{(n)}}}{\sqrt{n} \Delta_k}) ) - \frac{1}{2} (\bm \delta_{kX}^0 - \bm \delta_{tX}^0)' \bm C_X^0 (\bm \delta_{kX}^0 - \bm \delta_{tX}^0) \right]^2} \nonumber   \\
& ~~~~~~~~~~~~~~~~~~~~~~~~~~~~~~~~ + \frac{\Delta_k^2 [ O_p(\frac{b_k^{(n)}}{\Delta_k}) + O_p(d_n) ]}{\left[\epsilon - \Delta^2 ( O_p(\frac{b_k^{(n)}}{\Delta_k}) + O_p(d_n) + O_p(\frac{\sqrt{S_{h;p} q_k^{(n)}}}{\sqrt{n} \Delta_k}) ) \right]^2} \nonumber \\
\leq & \frac{\Delta^2 O_p(\xi_{n;k})}{[\epsilon - \Delta^2 O_p(\xi_{n;k}) - \frac{1}{2} (\bm \delta_{kX}^0 - \bm \delta_{tX}^0)' \bm C_X^0 (\bm \delta_{kX}^0 - \bm \delta_{tX}^0)]^2} + \frac{\Delta^2 O_p(\xi_{n;k})}{[\epsilon - \Delta^2 O_p(\xi_{n;k})]^2}.
\end{align}
By condition (C6), the first term of \eqref{eq:7} converges to 0 in probability.
Pick $\epsilon = C \xi_{n;k}^\alpha$, where $0< \alpha < 1/2$ with a positive constant $C$, then
\begin{align*}
& ~~~~ \Pr(|\hat{\vartheta}_k - \vartheta_k| \geq \frac{\epsilon}{2}~ \mbox{for some } k) \\
& \leq  \frac{\Delta^2 O_p(\xi_{n;k})}{[\epsilon - \Delta^2 O_p(\xi_{n;k}) - \frac{1}{2} (\bm \delta_{kX}^0 - \bm \delta_{tX}^0)' \bm C_X^0 (\bm \delta_{kX}^0 - \bm \delta_{tX}^0)]^2} + \frac{\Delta^2 O_p(\xi_{n;k}^{1 - 2 \alpha})}{[C - \Delta^2 O_p(\xi_{n;k}^{1 - \alpha})]^2} \\
& \stackrel{P}{\rightarrow} 0.
\end{align*}
\end{proof}

\begin{proof}{\textbf{Proof of Theorem \ref{theory:consistent}}.}{}

The conditional misclassification rate is
\begin{align*}
R_{PROP}(\mathcal{T}) &= \frac{1}{2} \sum_{k=1}^2 \Phi \left(\frac{(-1)^k \hat{\bm \delta}'_{2X} \hat{\bm C}_X (\bm \mu_{kX}^0 - \hat{\bm \mu}_{kX}) - \hat{\bm \delta}'_{2X} \hat{\bm C}_X (\hat{\bm \mu}_{1X} - \hat{\bm \mu}_{2X}) / 2 }{\sqrt{\hat{\bm \delta}'_{2X} \hat{\bm C}_X \bm \Sigma_X^0 \hat{\bm C}_X \hat{\bm \delta}_{2X} }} \right) \\
&= \frac{1}{2} \sum_{k=1}^2 \Phi \left(\frac{(-1)^k \hat{\bm \delta}'_{2X} \hat{\bm C}_X (\bm \mu_{kX}^0 - \hat{\bm \mu}_{kX}) - \hat{\bm \delta}'_{2X} \hat{\bm C}_X \hat{\bm \delta}_{2X} }{\sqrt{\hat{\bm \delta}'_{2X} \hat{\bm C}_X \bm \Sigma_X^0 \hat{\bm C}_X \hat{\bm \delta}_{2X} }} \right).
\end{align*}
By the result \eqref{conrateC}, we have
\begin{equation*}
\hat{\bm \delta}'_{2X} \hat{\bm C}_X \bm \Sigma_X^0 \hat{\bm C}_X \hat{\bm \delta}_{2X} = \hat{\bm \delta}'_{2X} \hat{\bm C}_X \hat{\bm \delta}_{2X} [1 + O_p(d_n)] = \hat{\bm \delta}'_{2X} \bm C_X^0 \hat{\bm \delta}_{2X} [1 + O_p(d_n)].
\end{equation*}
From the result \eqref{conratedelta}, together with $E[ (\bm \delta_{2X}^0)' \bm C_X^0 (\hat{\bm \delta}_{2X} - \bm \delta_{2X}^0)]^2 \leq \Delta_2^2 E[(\hat{\bm \delta}_{2X} - \bm \delta_{2X}^0)' \bm C_X^0 (\hat{\bm \delta}_{2X} - \bm \delta_{2X}^0)]$, it is easy to derive
\begin{align*}
\hat{\bm \delta}'_{2X} \bm C_X^0 \hat{\bm \delta}_{2X} &= (\bm \delta_{2X}^0)' \bm C_X^0 \bm \delta_{2X}^0 + (\bm \delta_{2X}^0)' \bm C_X^0 (\hat{\bm \delta}_{2X} - \bm \delta_{2X}^0) + (\hat{\bm \delta}_{2X} - \bm \delta_{2X}^0)' \bm C_X^0 (\hat{\bm \delta}_{2X} - \bm \delta_{2X}^0) \\
&= \Delta_2^2 + O_p(b_2^{(n)} \Delta_2) + O_p((b_2^{(n)})^2) \\
&= \Delta_2^2 [1 + O_p(\frac{b_2^{(n)}}{\Delta_2}) + O_p(\frac{(b_2^{(n)})^2}{\Delta_2^2})] \\
&= \Delta_2^2 [1 + O_p(\frac{b_2^{(n)}}{\Delta_2})].
\end{align*}
Hence we have
\begin{equation*}
\hat{\bm \delta}'_{2X} \hat{\bm C}_X \hat{\bm \delta}_{2X}
= \hat{\bm \delta}'_{2X} \bm C_X^0 \hat{\bm \delta}_{2X} [1 + O_p(d_n)]
= \Delta_2^2 [1 + O_p(\frac{b_2^{(n)}}{\Delta_2}) + O_p(d_n)].
\end{equation*}
Therefore, by Lemma \ref{lemma5}, we obtain
\begin{align*}
\frac{\hat{\bm \delta}'_{2X} \hat{\bm C}_X (\hat{\bm \mu}_{1X} - \bm \mu_{1X}^0) - \hat{\bm \delta}'_{2X} \hat{\bm C}_X \hat{\bm \delta}_{2X} }{\sqrt{\hat{\bm \delta}'_{2X} \hat{\bm C}_X \bm \Sigma_X^0 \hat{\bm C}_X \hat{\bm \delta}_{2X} }}
&= \frac{O_p(\sqrt{\frac{q_2^{(n)}}{n}}) + O_p(\sqrt{\frac{S_{h;p} q_2^{(n)}}{n}}) - \sqrt{\hat{\bm \delta}'_{2X} \hat{\bm C}_X \hat{\bm \delta}_{2X}}}{\sqrt{1+O_p(d_n)}}   \\
&= -\frac{\Delta_2}{2} \frac{\sqrt{1 + O_p(\frac{b_2^{(n)}}{\Delta_2}) + O_p(d_n)}}{\sqrt{1 + O_p(d_n)}} + \frac{O_p(\sqrt{\frac{S_{h;p} q_2^{(n)}}{n}})}{\sqrt{1 + O_p(d_n)}}    \\
&= -\frac{\Delta_2}{2} [1 + O_p(\frac{b_2^{(n)}}{\Delta_2}) + O_p(d_n)] + O_p(\sqrt{\frac{S_{h;p} q_2^{(n)}}{n}})   \\
&= -\frac{\Delta_2}{2} [1 + O_p(\frac{b_2^{(n)}}{\Delta_2}) + O_p(d_n) + O_p(\frac{\sqrt{S_{h;p} q_2^{(n)}}}{\sqrt{n} \Delta_2})]  \\
&= -\frac{\Delta_2}{2} [1 + O_p(\xi_n)].
\end{align*}
Similarly,  we have
\begin{equation*}
\frac{\hat{\bm \delta}'_{2X} \hat{\bm C}_X (\bm \mu_{2X}^0 - \hat{\bm \mu}_{2X}) - \hat{\bm \delta}'_{2X} \hat{\bm C}_X \hat{\bm \delta}_{2X} }{\sqrt{\hat{\bm \delta}'_{2X} \hat{\bm C}_X \bm \Sigma_X^0 \hat{\bm C}_X \hat{\bm \delta}_{2X} }} = -\frac{\Delta_2}{2} [1 + O_p(\xi_n)],
\end{equation*}
which proves theory.
\end{proof}

To establish the theoretical results in Theorem \ref{theory:supp}, we need a lemma from \cite{Shao2011Sparse}.
We state it here for completeness, and then prove Theorem \ref{theory:supp}.

\begin{lemma}{}{\label{lemma1}}
Let $a_n^{(1)}$ and $a_n^{(2)}$ be two sequences of positive numbers such that
$a_n^{(1)} \rightarrow \infty$ and $a_n^{(2)} \rightarrow 0$ as $n \rightarrow \infty$.
If $\lim\limits_{n \rightarrow \infty} a_n^{(1)} a_n^{(2)} = \rho$, where $\rho$ may be 0, positive, or $\infty$, then
\begin{equation*}
\lim_{n \rightarrow \infty} \frac{\Phi(-\sqrt{a_n^{(1)}} (1-a_n^{(2)}) )}{\Phi(-\sqrt{a_n^{(1)}})} = e^{\rho}.
\end{equation*}
\end{lemma}

\begin{proof}{}{}
See the proof of Lemma 1 in \cite{Shao2011Sparse}.
\end{proof}

\begin{proof}{\textbf{Proof of Theorem \ref{theory:supp}}.}{}

(1) Let $\phi$ be the density function of $N(0,1)$.
By the mean value theorem,
\begin{equation*}
R_{PROP}(\mathcal{T}) - R_{Bayes} = \Phi(-\frac{\Delta_2}{2} [1 + O_p(\xi_n)]) - \Phi(-\frac{\Delta_2}{2}) = - \phi(\tau_n) \frac{\Delta_2}{2} O_p(\xi_n),
\end{equation*}
where $\tau_n$ is between $-\frac{\Delta_2}{2}$ and $-\frac{\Delta_2}{2} [1 + O_p(\xi_n)]$. Since $\Delta_2$ is bounded, then $R_{Bayes}$ is bounded away from 0.
Hence,
\begin{equation*}
\frac{R_{PROP}(\mathcal{T})}{R_{Bayes}} - 1 = - \frac{\Delta_2}{2} \frac{\phi(\tau_n)}{R_{Bayes}} O_p(\xi_n) = O_p(\xi_n).
\end{equation*}

(2) When $\Delta_2 \rightarrow \infty$, we have $R_{PROP}(\mathcal{T}) \stackrel{P}{\rightarrow} 0$.
This, together with $\lim\limits_{\Delta_2 \rightarrow \infty} R_{Bayes} = 0$,
proves (2).

(3) The conditions $\Delta_2 \rightarrow \infty$, $\lim\limits_{n \rightarrow \infty} \xi_n \Delta_2^2 = 0$, together with Lemma \ref{lemma1}, prove that $R_{PROP}(\mathcal{T}) / R_{Bayes} \stackrel{P}{\rightarrow} 1$.
\end{proof}

\begin{proof}{\textbf{Proof of Theorem \ref{theory:y}}.}{}

Define $r_{ik}=\Pr(\hat{Z}=i|Z=k)$ for $i, k = 1, 2, \ldots, K$.
Let $R$ be the misclassification error for a classifier, it is then calculated via
\begin{align}\label{eq:R}
R = \sum_{k=1}^K \Pr(Z=k) \Pr(\hat{Z} \neq k|Z=k) = \sum_{k=1}^K \left( \pi_k \sum_{i \neq k}^K r_{ik} \right).
\end{align}
Now we derive an upper bound of $(\hat{y} - y)^2$.
Since it is random, we focus on the average, i.e.,
\begin{align}\label{eq:9}
\E[(\hat{y}-y)^2|\bm x, \mathcal{T}] = \E_y \E_{\hat{y}|\mathcal{T}} \left[(\hat{y}-y)^2|\bm x, \mathcal{T}\right].
\end{align}
To simplify the notation, we omit $\bm x$ and $\mathcal{T}$ from the right of the conditional sign and write it as $\E_y\E_{\hat{y}|\mathcal{T}}\left[(\hat{y}-y)^2\right]$.
Then Equation \eqref{eq:9} becomes
\begin{align*}
\E[(\hat{y}-y)^2] = \E_y \E_{\hat{y}|\mathcal{T}}\left[(\hat{y}-y)^2 \right] &= \E_Z (\E_{y|Z} \E_{\hat{y}|\mathcal{T}}\left[(\hat{y}-y)^2 \right] | Z) \\
&= \sum_{k=1}^K \pi_k \E_{y|Z=k} \left( \sum_{i=1}^K r_{ik} (\hat{y}_i-y)^2 | Z = k \right)
\end{align*}
Next, we derive
\begin{align*}
&\E_{y|Z=1}\left[a_1(\hat{y}_1-y)^2|Z=1\right]  \\
=& \E_{y|Z=1}\left[a_1(\hat{y}_1-\E(y|Z=1)+\E(y|Z=1)-y)^2|Z=1\right]\\
=&\E_{y|Z=1}\left[a_1(\hat{y}_1-\E(y|Z=1))^2|Z=1\right]+\E_{y|Z=1}\left[a_1(y-\E(y|Z=1))^2|Z=1\right]\\
=& a_1(\hat{y}_1-\E(y|Z=1))^2+a_1 \mbox{Var}(y|Z=1).
\end{align*}
Similarly, we have
\begin{align*}
\E_{y|Z=k} \left[ c (\hat{y}_i - y)^2|Z=k \right]
= c(\hat{y}_i - \E(y|Z=k))^2 + c \mbox{Var}(y|Z=k)
\end{align*}
for $c > 0$ and $i, k = 1,2,\ldots,K$.
As a result, Equation \eqref{eq:9} is decomposed as
\begin{align*}
\E[(\hat{y}-y)^2|\bm x, \mathcal{T}] &= \sum_{k=1}^K \pi_k \left[ \sum_{i=1}^K r_{ik} (\hat{y}_i - \E(y|Z=k))^2 + \sum_{i=1}^K r_{ik} \mbox{Var}(y|Z=k) \right] \\
&= \sum_{k=1}^K \sum_{i=1}^K \pi_k r_{ik} (\hat{y}_i - \E(y|Z=k))^2 + (\sigma_y^2 - \bm \Sigma_{Xy}' \bm \Sigma_{X}^{-1} \bm \Sigma_{Xy}) \sum_{k=1}^K \sum_{i=1}^K \pi_k r_{ik} \\
&= \sum_{k=1}^K \sum_{i=1}^K \pi_k r_{ik} (\hat{y}_i - \E(y|Z=k))^2 + (\sigma_y^2 - \bm \Sigma_{Xy}' \bm \Sigma_{X}^{-1} \bm \Sigma_{Xy}),
\end{align*}
where the second equality applies $\mbox{Var}(y|Z=k) = \sigma_y^2 - \bm \Sigma_{Xy}' \bm \Sigma_{X}^{-1} \bm \Sigma_{Xy}$,
and the third equality uses the fact $\sum_{k=1}^K \sum_{i=1}^K \pi_k r_{ik} = 1$ based on the definition of $r_{ik}$.
Now we tackle with each term of
\begin{align*}
&[\hat{y}_k - \E(y|Z=k)]^2    \\
=&\left[(\hat{\mu}_{ky}-\mu_{ky})+\left(\hat{\bm \Sigma}_{Xy}' \hat{\bm \Sigma}_{X}^{-1}-\bm \Sigma_{Xy}' \bm \Sigma_{X}^{-1}\right) \bm x - \left(\hat{\bm \Sigma}_{Xy}' \hat{\bm \Sigma}_{X}^{-1} \hat{\bm \mu}_{kX}-\bm \Sigma_{Xy}' \bm \Sigma_{X}^{-1} \bm \mu_{kX} \right) \right]^2 \\
\end{align*}
and
\begin{align*}
&\left[\hat{y}_k-\E(y|Z=k')\right]^2  \\
=&\left[(\hat{\mu}_{ky}-\mu_{k'y})+\left(\hat{\bm \Sigma}_{Xy}' \hat{\bm \Sigma}_{X}^{-1}-\bm \Sigma_{Xy}' \bm \Sigma_{X}^{-1} \right)\bm x - \left(\hat{\bm \Sigma}_{Xy}' \hat{\bm \Sigma}_{X}^{-1} \hat{\bm \mu}_{kX}-\bm \Sigma_{Xy}' \bm \Sigma_{X}^{-1}\bm \mu_{k'X}\right)\right]^2\\
=&\left[(\hat{\mu}_{ky}-\mu_{ky})+\left(\hat{\bm \Sigma}_{Xy}' \hat{\bm \Sigma}_{X}^{-1} - \bm \Sigma_{Xy}' \bm \Sigma_{X}^{-1} \right) \bm x - \left(\hat{\bm \Sigma}_{Xy}' \hat{\bm \Sigma}_{X}^{-1} \hat{\bm \mu}_{kX}-\bm \Sigma_{Xy}' \bm \Sigma_{X}^{-1}\bm \mu_{kX}\right)\right.\\
&+\left.(\mu_{ky}-\mu_{k'y})-\left(\bm \Sigma_{Xy}' \bm \Sigma_{X}^{-1} \bm \mu_{kX}-\bm \Sigma_{Xy}' \bm \Sigma_{X}^{-1}\bm \mu_{k'X}\right)\right]^2,~\mbox{for}~k \neq k'.
\end{align*}
For $k = 1,2,\ldots,K$ and $k \neq k'$, define the following terms
\begin{align*}
b_{ky} &=\hat{\mu}_{ky}-\mu_{ky},\\
D_{kk'} &=\E(y|Z=k)-\E(y|Z=k')   \\
&=(\mu_{ky}-\mu_{k'y})-\left(\bm \Sigma_{Xy}' \bm \Sigma_{X}^{-1}\bm \mu_{kX}-\bm \Sigma_{Xy}' \bm \Sigma_{X}^{-1}\bm \mu_{k'X}\right),\\
\bm h &=\left(\hat{\bm \Sigma}_{Xy}' \hat{\bm \Sigma}_{X}^{-1}-\bm \Sigma_{Xy}' \bm \Sigma_{X}^{-1}\right)',\\
d_k &=\hat{\bm \Sigma}_{Xy}' \hat{\bm \Sigma}_{X}^{-1}\hat{\bm \mu}_{kX}-\bm \Sigma_{Xy}' \bm \Sigma_{X}^{-1}\bm \mu_{kX}.
\end{align*}
Therefore, we obtain
\begin{align*}
&\E[(\hat{y}-y)^2|\bm x, \mathcal{T}] \\
=& \sum_{i=1}^K \pi_i r_{ii} (b_{iy}+\bm h' \bm x-d_i)^2 + \sum_{k=1}^K \sum_{i \neq k}^K \pi_k r_{ik} (b_{iy}+\bm h' \bm x-d_i + D_{ik})^2 + (\sigma_y^2 - \bm \Sigma_{Xy}' \bm \Sigma_{X}^{-1} \bm \Sigma_{Xy}) \\
=& \sum_{i=1}^K \pi_i r_{ii} (b_{iy}+\bm h' \bm x-d_i)^2 + \sum_{k=1}^K \sum_{i \neq k}^K \pi_k r_{ik} (b_{iy}+\bm h' \bm x-d_i)^2 + \sum_{k=1}^K \sum_{i \neq k}^K \pi_k r_{ik} D_{ik}^2 \\
&~~~~~~~~~~~~~~~~~~~~~~~~~~~~~~~~+ \sum_{k=1}^K \sum_{i \neq k}^K 2 \pi_k r_{ik} (b_{iy}+\bm h' \bm x-d_i) D_{ik} + (\sigma_y^2 - \bm \Sigma_{Xy}' \bm \Sigma_{X}^{-1} \bm \Sigma_{Xy})  \\
=& \mathbb{M} + \sum_{k=1}^K \sum_{i \neq k}^K \pi_k r_{ik} D_{ik}^2 +  (\sigma_y^2 - \bm \Sigma_{Xy}' \bm \Sigma_{X}^{-1} \bm \Sigma_{Xy}),
\end{align*}
where $$\mathbb{M} = \sum_{k=1}^K \sum_{i=1}^K \pi_k r_{ik} (b_{iy}+\bm h' \bm x-d_i)^2 + \sum_{k=1}^K \sum_{i \neq k}^K 2 \pi_k r_{ik} (b_{iy}+\bm h' \bm x-d_i) D_{ik}.$$
Now if the classification of $Z$ is based on the known distribution, the misclassification rate $R$ is $R_{Bayes}$.
For $i,k = 1, 2, \ldots, K$, let $r_{ik}^{B} = \Pr(\hat{Z}=i|Z=k)$ represent the corresponding $r_{ik}$ with $\hat{Z}$ obtained from Bayes rule.
Similarly, let symbol $r_{ik}^{P}$ represent the corresponding $r_{ik}$ with $\hat{Z}$ obtained from the proposed model.
Denote by $\mathbb{M}_{PROP}$ the corresponding value of $\mathbb{M}$ computed from the proposed model.
Note that the value of $\mathbb{M}$ computed from Bayes rule is equal to 0.
Then we have
\begin{align*}
\E[(\hat{y}^{P}-y)^2|\bm x, \mathcal{T}] - \E[(\hat{y}^{B}-y)^2|\bm x, \mathcal{T}]
&= \mathbb{M}_{PROP} + \sum_{k=1}^K \sum_{i \neq k}^K (\pi_k r_{ik}^{P} - \pi_k r_{ik}^{B}) D_{ik}^2    \\
&\leq  \mathbb{M}_{PROP} + [R_{PROP}(\mathcal{T}) - R_{Bayes}] D_{max}^2,
\end{align*}
where $D_{max}^2$ = max $\{ D_{kk'}^2 \}$, and the last inequality uses Equation \eqref{eq:R}.
By conditions in Theorem \ref{theory:multiclass}, $\E_{\bm x}(\mathbb{M}_{PROP}) \stackrel{P}{\rightarrow} 0$ as $n \rightarrow \infty$.
Consequently, we have
\begin{align*}
MSE_{PROP} - MSE_{Bayes} &= \E[(\hat{y}^{P}-y)^2|\mathcal{T}] - \E[(\hat{y}^{B}-y)^2|\mathcal{T}] \\
&= \E_{\bm x}\E[(\hat{y}^{P}-y)^2|\bm x, \mathcal{T}] - \E_{\bm x}\E[(\hat{y}^{B}-y)^2|\bm x, \mathcal{T}]  \\
& \leq \E_{\bm x}(\mathbb{M}_{PROP}) + [R_{PROP}(\mathcal{T}) - R_{Bayes}] D_{max}^2 \\
& \stackrel{P}{\rightarrow} 0.
\end{align*}

\end{proof}


\end{document}